\documentclass[11pt]{article}
\usepackage[T1]{fontenc}
\usepackage{geometry}
\usepackage[round,authoryear]{natbib}
\usepackage{float}
\usepackage{tcolorbox}
\usepackage{arydshln}
\usepackage{bm}
\usepackage{setspace}
\usepackage{amsmath}
\usepackage{tikz}
\usepackage{diagbox}
\usepackage{bbm}
\usepackage{adjustbox}
\usepackage{graphics}
\usepackage{xcolor}
\usepackage{arydshln}
\usepackage{multirow}
\usepackage{stackengine}

\stackMath
\newcommand\tsup[2][2]{%
	\def\useanchorwidth{T}%
	\ifnum#1>1%
	\stackon[-.5pt]{\tsup[\numexpr#1-1\relax]{#2}}{\scriptscriptstyle\sim}%
	\else%
	\stackon[.5pt]{#2}{\scriptscriptstyle\sim}%
	\fi%
}
\usepackage{arydshln}
\usepackage{amsmath}
\usepackage{amsthm}
\usepackage{amssymb}

\usepackage{booktabs,caption}
\usepackage{threeparttable}
\usepackage{graphicx}
\usepackage{color}
\usepackage{bbm}
\makeatletter
\theoremstyle{plain}
\newtheorem{thm}{\protect\theoremname}
\theoremstyle{plain}
\newtheorem{assumption}{\protect\assumptionname}[section]
\theoremstyle{plain}
\newtheorem{prop}{\protect\propositionname}[section]
\theoremstyle{plain}
\newtheorem{lem}{\protect\lemmaname}[section]
\theoremstyle{plain}

\theoremstyle{plain}

\newtheorem{Corollary}{Corollary}[section]

\makeatother

\usepackage[english]{babel}
\providecommand{\assumptionname}{Assumption}
\providecommand{\lemmaname}{Lemma}
\providecommand{\propositionname}{Proposition}
\providecommand{\theoremname}{Theorem}

\usepackage{thmtools}
\renewcommand\thmcontinues[1]{Continued}
\usepackage{authblk}
\numberwithin{equation}{section}

\usepackage[hidelinks,colorlinks=true, citecolor=blue, linkcolor=blue]{hyperref} 
\usepackage{xr}
\title{Treatment Effects in the Regression Discontinuity Model with Counterfactual Cutoff and Distorted Running Variables}
\author[1]{Moyu Liao \thanks{First draft: Nov, 2025. School of Economics, University of Sydney. Email: moyu.liao@sydney.edu.au.} }
\affil[1]{University of Sydney}
\date{\today}

\begin{document}
	\maketitle

\begin{abstract}
We develop a new framework for evaluating the total policy effect in regression discontinuity designs (RDD), incorporating both the direct effect of treatment on outcomes and the indirect effect arising from distortions in the running variable when treatment becomes available. Our identification strategy combines a conditional parallel trend assumption to recover untreated potential outcomes with a local invariance assumption that characterizes how the running variable responds to counterfactual policy cutoffs. These components allow us to identify and estimate counterfactual treatment effects for any proposed threshold. We construct a nonparametric estimator for the total effect, derive its asymptotic distribution, and propose bootstrap inference procedures. Finally, we apply our framework to the Italian Domestic Stability Pact, where population-based fiscal rules generate both behavioral responses and running-variable distortions.
\end{abstract}
	
\section{Introduction}

Regression Discontinuity Designs (RDD) have become one of the most influential empirical strategies in applied economics, in large part due to their transparent identification of causal effects in settings where treatment assignment is determined by a cutoff rule. \cite{lee2008randomized} provides a foundational interpretation of the sharp RDD by showing that, under continuity of potential outcomes at the threshold, the discontinuity in observed outcomes identifies a \emph{local average treatment effect} (LATE) for units at the margin of indifference. Importantly, Lee's framework allows the running variable to be influenced by individuals' anticipation of treatment, as long as such sorting does not generate a discontinuity in potential outcomes at the cutoff. In this sense, the RDD estimand can be understood as a characteristics–weighted average treatment effect for units whose underlying attributes place them near the policy threshold.

While this framework is extremely powerful, it is limited in several important dimensions when the research or policy question extends beyond the immediate neighborhood of the cutoff. First, the classical RDD focuses exclusively on the \emph{direct} effect of treatment on outcomes by conditioning on the running variable, without accounting for the possibility that the running variable itself may be \emph{endogenously distorted} by the availability or anticipation of treatment. Second, identification is restricted to a \emph{local} parameter at the cutoff, whereas policy makers may be interested in causal effects for individuals \emph{away} from the boundary. Third, standard RDD methods are tied to the \emph{observed} cutoff, leaving open the policy-relevant question of how treatment effects would differ under a \emph{counterfactual threshold}.

To address these three limitations, we develop a new framework that augments the standard RDD setup with a two-period, two-group structure. We consider periods $t\in\{0,1\}$ and groups $Z_i\in\{0,1\}$. For individuals with $Z_i=0$, treatment is never available in either period. For individuals with $Z_i=1$, treatment is unavailable in period $t=0$ but becomes available in period $t=1$ depending on their running variable. We introduce a potential-outcome framework for the running variable,$
R_{it} \;=\; R_{it}^0(1-Z_i) \;+\; 
Z_i\Big[ R_{it}^0\mathbbm{1}(t=0) + R_{it}^1\mathbbm{1}(t=1)\Big],$
which allows the running variable to be distorted from $R_{it}^0$ to $R_{it}^1$ when treatment becomes available. By defining the \emph{total policy effect} on individual $i$ in period $t=1$ as $
Y_{i1}^1(R_{i1}^1)-Y_{i1}^0(R_{i1}^0),$
our framework explicitly incorporates both the direct treatment effect and the indirect effect arising from the endogenous distortion of the running variable. This directly resolves the first limitation of the prior literature.

To address the second limitation—extrapolating treatment effects away from the cutoff—we leverage the $Z_i=0$ group as a control group. Specifically, we impose a conditional parallel trend assumption on untreated potential outcomes between the two groups, which allows us to recover the counterfactual untreated outcome $Y_{i1}^0(R_{i1}^0)$ for the $Z_i=1$ units. This strategy differs fundamentally from the two main existing extrapolation approaches: the local policy invariance method of \cite{dong2015identifying}, which extrapolates using local derivatives near the cutoff, and the conditional-independence-based extrapolation of \cite{angrist2015wanna}, which conditions on covariates to remove dependence between potential outcomes and the forcing variable. In contrast, our approach uses \emph{overtime variation} and \emph{between-group comparisons} to build the necessary counterfactual.

To address the final limitation—evaluating treatment effects under counterfactual cutoffs—we assume that the distortion of the running variable under a counterfactual policy threshold exhibits a translation-invariant structure relative to the observed distortion. Combined with the conditional parallel trend assumption, this allows us to predict both the counterfactual distribution of the distorted running variable and the implied counterfactual treatment effects at any proposed cutoff.

Building on these identification ideas, we construct a nonparametric estimator of the counterfactual total policy effect by combining kernel estimators for (i) the conditional mean of untreated potential outcomes, (ii) the conditional density of the running variable distortion, and (iii) the conditional mean of realized treated outcomes. Following the procedure outlined in Section~\ref{section: Asymptotics}, we derive the asymptotic linear expansion of the estimator and show that, despite the multiple nonparametric components, the estimator converges at the $\sqrt{n}$ rate. This fast rate arises because the estimand is an \emph{average} counterfactual policy effect that depends only on the treated subsample, and because undersmoothing renders the nonparametric bias asymptotically negligible. We evaluate the finite-sample performance through simulations and find that, due to the persistence of finite-sample bias, a two-scale bias-reduction estimator performs substantially better. Both empirical and $t$-bootstrap procedures deliver accurate confidence interval coverage in moderately sized samples.

We illustrate the usefulness of our approach by applying it to the setting of \cite{grembi2016fiscal}, who study the effect of Italy's Domestic Stability Pact (DSP) rules on municipal fiscal behavior. In their setting, municipalities below a population threshold faced different fiscal rules than those above it, and population itself may respond endogenously to fiscal incentives. Section \ref{section: Empirical} describes the institutional context and data. Our result shows that DSP decreases the municipal deficit for counterfactual cutoff set above 4,600 and can increase deficit if set below 4,600.

\paragraph{Literature} 
This paper contributes to the literature on RDD by extending the identification and estimation of treatment effects beyond the observed cutoff. \cite{lee2008randomized} establishes the modern interpretation of the RDD as identifying a local average treatment effect under continuity of potential outcomes. \cite{dong2015identifying} show how marginal changes in the policy threshold can identify the marginal threshold treatment effect under a local policy invariance assumption. \cite{angrist2015wanna} develop a conditional-independence-based strategy for extrapolating RDD effects away from the cutoff using covariates.

We also contribute to the growing literature combining Difference-in-Differences ideas with RDD designs. \cite{grembi2016fiscal} use a parallel trend structure to net out confounding mayoral wage policies that coincide with the same population cutoff, but their focus remains on the direct treatment effect rather than the full policy effect including the distortion of the running variable. Our framework incorporates both components and thus enriches the set of policy-relevant estimands available in RDD settings.

\section{Econometric Framework}\label{section: model}
Consider a two-period model $t=0,1$ and a regression discontinuity design that happens in period $t=1$. 
Let $Y_{it}$ be the observed outcome variable for individual $i$ in period $t$, let $R_{it}$ be the running variable that will determine the treatment status in period $t=1$. Let $Z_i\in\{0,1\}$ be a binary variable that determines whether the individual is subject to RDD in period $t=1$.  The treatment status $D_{it}=0$ for all individuals $i$ and $t=0$, and $D_{it}=\mathbbm{1}(R_{it}\ge c,Z_{i}=1)$ for period $t=1$.  The observed outcome variable is generated by the potential outcome model:
\begin{equation}\label{eq: potential outcome for Y}
	\begin{split}
		Y_{it}(R_{it})&=Y_{it}^1(R_{it})D_{it}+Y_{it}^0(R_{it})(1-D_{it}),\\
		R_{it}&= R_{it}^0(1-Z_i)+ Z_i\left[R_{it}^0 \mathbbm{1}(t=0)+ R_{it}^1 \mathbbm{1}(t=1)\right].
	\end{split}
\end{equation}
The running variable $R_{it}$ can also be influenced by the treatment due to the firm's endogenous behavior of choosing the running variable. The econometrician cannot observe $Y_{it}^1(R_{it})$ and $Y_{it}^0(R_{it})$ simultaneously and the potential outcome can depend on $R_{it}$. The running variable can also be influenced by the potential to be treated since individuals may strategically choose the running variable when the treatment is potentially available to them. As a result, for the control group $Z_i=0$, we always observe the policy undistorted potential running variable $R_{it}^0$, while for the $Z_i=1$ group, we observe $R_{it}^0$ for $t=0$ period and the distorted potential running variable $R_{it}^1$ for the $t=1$ period.  The econometrician may also observe a set of covariates $X_{it}$, but we leave the discussion to extensions. We also impose the following exogeneity condition for the potential outcome. The potential outcome framework in \eqref{eq: potential outcome for Y} is general: If we impose $E[Y_{it}^d(R_{it})|X_{it}=x]\equiv E[Y_{it}^d|X_{it}=x]$ for $d=0,1,$ then we get the model in \cite{angrist2015wanna}. 
\begin{assumption}\label{assumption: conditional exogeneity of the potential outcome }
	The potential outcome $Y_{i1}^0(R_{i1}^0)$ in period $t=1$ is exogenous to the distorted running variable given the undistorted running variable $R^0_{it}$, i.e.,
	\[
	Y_{i1}^0(R_{i1}^0)\perp D_{it},R_{i1}^1,R_{i0}^1 \big| Z_i,R_{i0}^0, R_{i1}^0.
	\]
\end{assumption}

The literature has focused on the local average treatment effect:
\[
S(r_1,r_2)=E[Y_{it}^1(R_{it})-Y_{it}^0(R_{it})|Z_i=1, c=r_1, R_{it}=r_2,t=1],
\]
which is the effect of treatment with treatment cutoff $c=r_1$ if an individual's running variable is fixed at $r_2>r_1$  \citep{dong2015identifying}. Under suitable continuity conditions, we can identify the local treatment effect at the treatment cutoff:
\begin{equation}
	S(c,c)=_{i.d.} E[Y_{it}|t=1,Z_i=1,R_{it}=c]- \lim_{\varepsilon\rightarrow 0^+}  E[Y_{it}|t=1,Z_i=1,R_{it}=c-\varepsilon].
\end{equation}

The local average treatment effect is not the only object that is interesting to the policy makers: it does not consider the policy effect via the distortion of the potential running variable. To take into account the distortion of running variables, we can also define the local effects of policy
\[
T(r_1,r_2)=E[Y_{i1}^1(R^1_{i1})-Y_{i1}^0(R_{i0}^0)|Z_i=1,c=r_1,R^1_{i1}=r_2].
\]
The total treatment effect with cutoff $c=r_1$ focuses on a person with post-treatment running variable $R_{i1}=r_2$, and compares the treated outcome $Y_{i1}^1(R^1_{i1})$ with its untreated potential outcome evaluated at the undistorted potential running variable $Y_{i1}^0(R_{i0}^0)$. The difference $T(r_1,r_2)-S(r_1,r_2)$ is the indirect effect of policy via the distortion of the running variable.

The policy maker is interested in both evaluating policy effect under the current policy cutoff $c$, as well as the counterfactual policy effect when the policy cutoff is chosen at a different value $c^{cf}>c$. When evaluating the treatment effect at each cutoff level, the policy maker also wants to know the local direct and indirect effect on the treated unit, i.e.,  $S(c',r)$, $T(c',r)$ for $r\ge c'$. For a policy maker who wants to know the counterfactual overall population treatment effect on the treated, we define the quantity
\[ATT(c^{cf})\equiv E[T(c^{cf},R_{it})|R_{it}\ge c^{cf}],\]
which is policy relevant and can inform the policy maker about the overall benefits of policy.

\subsection{Identification Strategy}	
To infer the treatment effect in a counterfactual cutoff location, we need to consider several steps: First, we use the over-time variation of $(Y_{it},R_{it})|Z_i=0$ to infer the trend of $Y_{it}^0(R_{it}^0)$. When we derive the trend of $Y_{it}^0(R_{it}^0)$, the potential outcome model \eqref{eq: potential outcome for Y} implies the correct conditional parallel trend assumption to impose.  Second, for $Z_i=1$ and $t=1$ period, we can identify the realized distributions of $(Y_{it}^0,R^1_{it})|Z_i=1,t=1,R_{it}<c$ and $(Y_{it}^1,R^1_{it})|Z_i=1,t=1,R_{it}\ge c$. The conditional distribution $R_{i1}^1|R_{i1}^0,Z_i=1$ informs us how the individuals select the distorted running variable given their past undistorted running variable. Last, to infer the treatment effect in the new cutoff place, we use the difference between the distribution $Y_{it}|Z_i=1,t=1$ and $Y_{it}^0|Z_{i}=1,t=1$ to inform us the treatment effect conditional on the running variables, and then impose additional translation invariance in the running variable distortion to inform us of the distortion of running variables under the new policy cutoff.
\paragraph{Identifying the counterfactual outcome for $D_{it}=1$} We first impose the conditional parallel trend assumption in the untreated potential outcome conditional on the running variable:
\begin{assumption}\label{assumption: conditional parallel trend assumption in running variable}
	(Conditional parallel trend in running variables) 
	\[
	\begin{split}
		\quad E[Y_{i1}^0(R_{i1}^0)-Y_{i0}^0(R_{i0}^0)|Z_{i}=1,R_{i0}^0=r_0]= E[Y_{i1}^0(R_{i1}^0)-Y_{i0}^0(R_{i0}^0)|Z_{i}=0, R_{i0}^0=r_0].
	\end{split}
	\]
\end{assumption}
We choose $R_{i0}^0$ as the conditioning variable in Assumption  \ref{assumption: conditional parallel trend assumption in running variable} based on two reasons: first, we cannot observe $R_{i1}^0$ for the $Z_i=1$ group and $R_{i1}^1$ for the $Z_i=0$ group, so it is not possible to add these running variables in the parallel trend assumption; second, as we will see later, the distortion of running variable can be characterized as the conditional distribution of $R_{i1}^1$ given $R_{i0}^0$, and conditioning on $R_{i0}^0$ will deliver the right conditional unobserved potential outcome expectation so that we can take into account the distortion of the running variables.

\begin{lem}\label{lem: parallel trend for Y^0 using running variables}
	Assumption \ref{assumption: conditional parallel trend assumption in running variable} identifies the conditional mean 
	$E[Y_{i1}^0(R^0_{i1})| R^0_{i0}=r_0, Z_i=1]$.
\end{lem}

\paragraph{Identifying the distortion of running variables} When individuals can potentially have access to the treatment, the running variable may be intentionally distorted to select into or out of treatment. The way that the running variables are distorted is characterized by the following conditional distribution:
\[
F(R_{i1}^1\le r_1 |R_{i0}^0=r_0,Z_i=1).
\]
To characterize the counterfactual distribution of the distorted running variables, we impose the following translation invariance in the selection behavior when the counterfactual cutoff is $c^{cf}>c$:
\begin{assumption}\label{assumption: selection behavior}
	(Translation invariance selection) Let $R^{1}_{i1}(c^{cf})$ be the counterfactual running variable if the policy cutoff is chosen at $c^{cf}\ge c$. We require the selection behavior of running variable to depend only on the distance to the cutoff point: Suppose the new cutoff is $c^{cf}=c+\Delta c$, then 
	\[
	Pr(R_{i1}^{1}(c^{cf})\le r_1+\Delta c |R_{i0}^0=r_0+\Delta c,Z_i=1)=F(R_{i1}^1\le r_1 |R_{i0}^0=r_0,Z_i=1)
	\]
\end{assumption}
Assumption \ref{assumption: selection behavior} has the following interpretation: For treated units, its distorted running variable $R_{i1}^1$ depends on its untreated potential running variable $R_{i0}^0$. When the cutoff for treatment is moved to a new location $c^{cf}=c+\Delta c$, the selection of $R_{it}^1(c^{cf})$ takes a parallel shift. 

\begin{lem}\label{lem: counterfactual distribution of R_{it} in Z_i=1}
	Assumption \ref{assumption: selection behavior} identifies the counterfactual distribution of $(R_{i0}^0,R_{it}^1(c^{cf}))$ conditional on the $Z_i=1$ group. As a result, the conditional distribution 
	$F(R_{i0}^0\le r_0|R^{1}_{i1}(c^{cf})=r)$ for the $Z_i=1$ group and the marginal distribution $F^{cf}(r)\equiv Pr(R_{i1}(c^{cf})\le r)$ are identified.
\end{lem}

\subsubsection*{Identifying the counterfactual treatment effect}
We notice that the counterfactual local average policy effect $T(c^{cf},r)$ is identified for all $r\ge c^{cf}\ge c$.

\begin{thm}\label{theorem: identify conditional total treatment effect}
	The counterfactual local average policy effect for the $Z_i=1$ group, $T(c^{cf},r)$, is identified as
	\begin{equation} \label{eq: local total policy effect}
		\begin{split}
			T(c^{cf},r)=E[Y_{i1}|Z_i=1,R_{i1}=r] -	\int_{r_0} E[Y_{i1}^0(R_{i1}^0)|Z_i=1,R_{i0}^0=r_0]dF(R_{i0}^0=r_0|R^{1}_{i1}(c^{cf})=r)
		\end{split}
	\end{equation}
	where $E[Y_{i1}^0(R_{i1}^0)|Z_i=1,R_{i0}^0=r_0] $ is identified in Lemma \ref{lem: parallel trend for Y^0 using running variables}, $F(R_{i0}^0=r_0|R^{1}_{i1}(c^{cf})=r)$ is identified in Lemma \ref{lem: counterfactual distribution of R_{it} in Z_i=1}.
\end{thm}

The \eqref{eq: local total policy effect} in Theorem \ref{theorem: identify conditional total treatment effect} summarizes our identification strategy: The $E[Y_{i1}|Z_i=1,R_{i1}=r]$ corresponds to the mean of treated potential outcome with the distorted running variable, the integrand $E[Y_{i1}^0(R_{i1}^0)|Z_i=1,R_{i0}^0=r_0]$ is the untreated potential outcome under undistorted running variable, and the distribution $F(R_{i0}^0=r_0|R^{1}_{i1}(c^{cf})=r)$ characterizes the distortion of the running variable in the counterfactual policy cutoff. 

Last, we can identify the treatment effect on the treated (ATT) using the counterfactual distribution of $R_{it}^1(c^{cf})$.
\begin{prop}
	For counterfactual cutoff $c^{cf}\ge c$, the counterfactual ATT is given by 
	\begin{equation}
		\label{eq: identified ATT }
		ATT(c^{cf})=\frac{ \int_{c^{cf}}^{\infty} T(c^{cf},r) dF^{cf}(r)}{1-F^{cf}(c^{cf})}.
	\end{equation}
\end{prop}
In the estimation procedure below, we will follow the identification formula in  \eqref{eq: local total policy effect} and \eqref{eq: identified ATT } to construct an estimator for the conditional mean and density components and assemble them to construct the final estimators for  \eqref{eq: local total policy effect} and \eqref{eq: identified ATT } .

\section{Estimation and Asymptotic Properties}\label{section: Asymptotics}

In this section, we outline the nonparametric kernel estimation procedures for the key conditional densities and expectations required to compute the counterfactual treatment effects discussed in Section 2. Let \( K(\cdot) \) be the Epanechnikov kernel function, let $K_h(\cdot)\equiv K(\cdot/h)$ denote the bandwidth normalized kernel function, and let \( h_n, b_n \) correspondingly denote the bandwidth sequences for the $(R_{i0},R_{i1})$.

\paragraph{Estimating the Marginal Distribution $f_0(r_0):=f(R_{i0}^0=r_0|Z_i=1)$} This marginal density is later used to infer the counterfactual distribution of $R_{it}^1$ with selection behavior. Note that, for $t=0$, $R_{i0}^0=R_{i0}$, so we can estimate $f_0(r_0)$ by:
\begin{equation} \label{eq: marginal density f(r_0)}
	\hat{f}(R_{i0}^0|Z_i=1)\equiv \widehat{f}_0(r_0) =
	\sum_{i: Z_i = 1} K_{{h}_n}(R_{i0} - r_0).
\end{equation}

\paragraph{Estimating the Conditional Mean \( m_0(r_0):=\mathbb{E}[Y_{i0}^0(R_{i0}^0) \mid R_{i0}^0 = r_0, Z_i = 1] \)}
This expectation is identified from the pre-treatment period \( t = 0 \) in the treated group by realizing that $Y_{i0}^0(R_{i0}^0)=Y_{i0}$ and $R_{i0}^0=r_0$. The corresponding estimator is:
\begin{equation}
	\widehat{m}_0(r_0)\equiv \widehat{\mathbb{E}}[Y_{i0}^0(R_{i0}^0) \mid R_{i0}^0 = r_0, Z_i=1] =
	\frac{ \sum_{i: Z_i = 1} K_{h_n}(R_{i0} - r_0) Y_{i0} }{
		\sum_{i: Z_i = 1} K_{h_n}(R_{i0} - r_0)
	}.
\end{equation}
Here \( h_n \) denotes the bandwidth for the running variable in the $Z_i=1$ group.

\paragraph{Estimating the Mean Difference \(g(r_0):=\mathbb{E}[Y_{i1}^0(R_{i1}^0) - Y_{i0}^0(R_{i0}^0) \mid R_{i0}^0 = r_0, Z_i = 0]\)}  This conditional mean difference quantity is identified from the $Z_i=0$ group by realizing that $Y_{i1}^0(R_{i1}^0)=Y_{i1}$ and $Y_{i0}^0(R_{i0}^0) =Y_{i0}$ for the $Z_i=0$ individuals. Define the difference \( \Delta Y_i = Y_{i1} - Y_{i0} \), and estimate:

\begin{equation}
	\widehat{g}(r_0)\equiv 	\widehat{\mathbb{E}}[\Delta Y_i \mid R_{i0} = r_0] =
	\frac{
		\sum_{i: Z_i = 0} K_{h_n}(R_{i0} - r_0) \cdot \Delta Y_i
	}{
		\sum_{i: Z_i = 0} K_{h_n}(R_{i0} - r_0) 
	}.
\end{equation}
Here \( h_n \) denotes the bandwidths for \( R_{i0} \).

\paragraph{Estimating the Conditional Density \( f_{1|0}(r_1|r_0):=f(R_{i1}^1 \mid R_{i0}^0, Z_i = 1) \)}
This density captures selection behavior in the treated group. Using observed running variable \( R_{i1} \) to replace \( R_{i1}^1 \) and $R_{i0}$ to replace $R_{i0}^0$ for the $Z_i=1$ group, we estimate:

\begin{equation}\label{eq: conditional density}
	\widehat{f}_{1|0}(r_1 \mid r_0) \equiv \widehat{f}(R_{i1}^1 \mid R_{i0}^0, Z_i = 1)	=
	\frac{
		\sum_{i: Z_i = 1} K_{b_n}(R_{i1} - r_1) K_{h_n}(R_{i0}^0 - r_0)
	}{
		\sum_{i: Z_i = 1} K_{h_n}(R_{i0}^0 - r_0)
	}.
\end{equation}
The marginal density of counterfactual running variable $R_{i1}^1(c^{cf})$ can be calculated as 
\[
\widehat{f}^{cf}(r_1)=\int_{r_0} \widehat{f}_{1|0}(r_1-\Delta c|r_0-\Delta c) \hat{f}_0(r_0) dr_0
\]

\paragraph{Estimating the Conditional Mean $m_1(r):=E[Y_{i1}|R_{i1}=r,Z_i=1]$} This object corresponds to the conditional mean $E[Y^1_{i1}|R^1_{i1}=r,Z_i=1]$ if $r\ge c$. Since for the $Z_i=1$ group, we only observe $Y_{i1}^1$ for $R_{i1}\ge c$, we have a boundary issue for the nonparametric estimation. Therefore, we need to use the boundary kernel estimation. Let $\widehat{m}_1(r)$ be the nonparametric estimator such that 
\begin{equation}
	\widehat m_1(r)
	=
	\begin{cases}
		\dfrac{\displaystyle
			\sum_{i:Z_i=1} 
			K\!\left(\dfrac{R_{i1}-r}{b_n}\right) Y_{i1}}
		{\displaystyle
			\sum_{i:Z_i=1} 
			K\!\left(\dfrac{R_{i1}-r}{b_n}\right)},
		& \text{if } r-c \ge \tau\,b_n,\\[3ex]
		\dfrac{\displaystyle
			\sum_{i:Z_i=1} 
			K^{\mathrm{bd}}_{r,b_n}\!\left(\dfrac{R_{i1}-r}{b_n}\right) Y_{i1}}
		{\displaystyle
			\sum_{i:Z_i=1} 
			K^{\mathrm{bd}}_{r,b_n}\!\left(\dfrac{R_{i1}-r}{b_n}\right)},
		& \text{if } c \le r-c < \tau\,b_n,
	\end{cases}
\end{equation}
where the boundary kernel \(K^{\mathrm{bd}}_{r,b_n}(u)\) adjusts the moments of the
base kernel \(K(u)\) on the truncated support \(\{u \ge (c-r)/b_n\}\):
\[
A_0(r)=\int_{(c-r)/b_n}^{\infty} K(u)\,du,\qquad
A_1(r)=\int_{(c-r)/b_n}^{\infty} u\,K(u)\,du,\qquad
A_2(r)=\int_{(c-r)/b_n}^{\infty} u^2\,K(u)\,du,
\]
\[
a(r)=\frac{A_2(r)}{A_0(r)A_2(r)-A_1(r)^2},\qquad
b(r)=-\,\frac{A_1(r)}{A_0(r)A_2(r)-A_1(r)^2},
\]
\[
K^{\mathrm{bd}}_{r,b_n}(u)
=
\bigl[a(r)+b(r)\,u\bigr]\,
K(u)\,
\mathbbm{1}\!\left\{u\ge \tfrac{c-r}{b_n}\right\}.
\]
Here \(c\), the policy cutoff value, is the lower boundary of the support of \(R_{i1}^1\),
and \(\tau>0\) (e.g.\ \(\tau=1\)) determines the width of the boundary region.

\paragraph{Estimating the Counterfactual Total Policy Effect} Given the nonparametric estimators in \eqref{eq: marginal density f(r_0)}- \eqref{eq: conditional density}, we estimate the counterfactual conditional total policy effect 
\begin{equation}
	\begin{split}
		\widehat{T(c^{cf},r)}=\widehat{m}_1(r)-\int_{r_0}  \left[ \widehat{g}(r_0)+\widehat{m}_0(r_0) \right] \hat{f}_{1|0}(r-\Delta c|r_0-\Delta c)\frac{\hat{f}_0(r_0)}{\hat{f}^{cf}(r)}dr_0.
	\end{split}
\end{equation}
For the overall treatment effect on the treated, we can estimate by  
\begin{equation*}
	\begin{split}
		\widehat{ATT}(c^{cf}) 
		&= \int_{r \ge c^{cf}} \widehat{T}(c^{cf},r) \widehat{f}^{cf}(r)\,dr/ (1-\widehat{F}^{cf}(c^{cf}) )\\
		&=  \frac{\int_{r \ge c^{cf}} \widehat{m}_1(r) \widehat{f}^{cf}(r)dr -\int_{r \ge c^{cf}} \int_{r_0}
			\big[\,\widehat{g}(r_0)+\widehat{m}_0(r_0)\,\big]\,
			\widehat{f}_{1\mid 0}(r-\Delta c \mid r_0-\Delta c)\,
			\widehat{f}_0(r_0)\,dr_0\,dr }{1-\widehat{F}^{cf}(c^{cf})}
	\end{split}
\end{equation*}
where $\widehat{F}^{cf}(c^{cf})=\int_{r\ge c^{cf}} \widehat{f}^{cf}(r) dr$. 
\subsection{Asymptotic Properties}

\begin{assumption}[Kernel Function]\label{ass:G3}
	The univariate kernel function $K(\cdot)$ is symmetric, bounded, and integrates to one, 
	and it is of order~2 (i.e.\ $\int K(u)\,du=1$, $\int u\,K(u)\,du=0$, and 
	$\int u^2 K(u)\,du<\infty$). 
	For multivariate estimation, a product kernel $K(u)K(v)$ is used.
\end{assumption}

\begin{assumption}[Design Density]\label{ass:G2}
	For each group $Z_i\in\{0,1\}$, the marginal and joint densities $f_{R_{i0}\mid Z_i}(r_0)$, $f_{R_{i1}\mid Z_i}(r_1)$ and $f_{R_{i0},R_{i1}\mid Z_i}(r_0,r_1)$ exist, are continuous, and are uniformly bounded and bounded away from zero on the compact supports of $R_{i0}$ and $R_{i1}$. 
\end{assumption}

\begin{assumption}[Smoothness of Structural Functions]\label{ass:smoothness}
	The functions 
	$f_{R_{i0}\mid Z_i=1}(r_0)$, 
	$m_0(r_0)=E[Y_{i0}\mid R_{i0}=r_0,Z_i=1]$, 
	$g(r_0)=E[\Delta Y_i\mid R_{i0}=r_0,Z_i=0]$, 
	$f_{R_{i1},R_{i0}\mid Z_i=1}(r_1,r_0)$, 
	and 
	$m_1(r_1)=E[Y_{i1}\mid R_{i1}=r_1,Z_i=1]$
	are twice continuously differentiable in their respective arguments, 
	with uniformly bounded first and second derivatives on compact supports.
\end{assumption}

\begin{assumption}[Balanced Sample Size]\label{ass: balanced sample}
	Let $n_1=\sum_{i} \mathbbm{1}(Z_i=1)$ and $n_0=\sum_{i} \mathbbm{1}(Z_i=0)$, with $n=n_1+n_0$. Then $\lim_{n\rightarrow \infty } n_1/n= \alpha\in (0,1)$.	
\end{assumption}

Assumption \ref{ass:G3}-\ref{ass:smoothness} are standard nonparametric estimation assumptions on the nonparametric estimation. The uniformly bounded density condition in Assumption \ref{ass:G2} is important to derive the uniform convergence of the conditional mean and conditional density to the population. Assumption \ref{ass:smoothness} allows us to derive the uniform bias from kernel estimation. Assumption \ref{ass: balanced sample} is later used to derive the asymptotic normality of the estimator, and shows that the estimation variation from both $Z_i=0$ and $Z_i=1$ group matters asymptotically, which suits the common empirical context that researchers will encounter.

\begin{thm}
	Under Assumptions \ref{ass:G3}-\ref{ass:smoothness}, when $h_n\asymp b_n\asymp n^{-1/5-\delta}$, for a small $\delta\ge0$ number, the estimator of counterfactual conditional total policy effect has the following linear expansion
	{\footnotesize
		\[
		\begin{split}
			&\sqrt{nb_n}\left[\widehat{T(c^{cf},r)}-T(c^{cf},r) \right]\\
			&=    \sqrt{nb_n}\frac{1}{f_1(r)}
			\sum_{i:Z_i=1} K_r^*\!\Big(\frac{R_{i1}-r}{b_n}\Big)\,\eta_{i1}\\
			&-  \frac{1}{\alpha\sqrt{nb_n}}\frac{\int_{r_0} f_{1|0}(r-\Delta c|r_0-\Delta c)f_0(r_0)dr_0}{f^{cf}(r)^2} \sum_{i:Z_i=1}  \frac{\Big[K_{b_n}(R_{i1}-(r-\Delta c))-E\left[K_{b_n}(R_{i1}-(r-\Delta c))|R_{i0}=r-\Delta c\right]\Big]\,f_0(R_{i0}+\Delta c )}{f_{0|1}(R_{i0})}  \\
			&-\frac{1}{\alpha\sqrt{nb_n}}\sum_{Z_i=1} \theta(R_{i0}+\Delta c)\frac{f_0(R_{i0}+\Delta c)}{f^{cf}(r)f_{0|1}(R_{i0})}  \left[K\left(\frac{R_{i1}-r+\Delta c}{b_n}\right)-E\left[K\left(\frac{R_{i1}-r+\Delta c}{b_n}\right)\mid R_{i0}\right]\right] \\
			&+O_p(n^{-2.5\delta}+n^{2/5-1/2-\delta/2} +   n^{-0.3+\delta}\log n  +n^{-\delta/2})\\
		\end{split}
		\] 
	}
	where $\eta_{i1}=Y_{i1}-E[Y_{i1}|R_{i1}=r,Z_i=1]$. 
\end{thm}
\begin{proof}
	Follows directly from Lemma \ref{lem:m1-bdry-unif} and \ref{lemma: linearization for the second term in T}.
\end{proof}

The asymptotic linear expansion shows several things: first, the optimal MSE for the $\widehat{T(c^{cf},r)}$ is $n^{-4/5}$ by choosing $\delta=0$. Note that even if our estimator involves conditioning on 2-dimensional variables $R_{i0},R_{i1}$, the convergence rate is the same as the standard 1-dimensional nonparametric kernel estimation, because we can undersmooth the bandwidth for $h_n$ and exploit the smoothing in the integration in the estimator; second, the linear expansion suggests the asymptotic normality of the estimator, and a fast multiplier bootstrap method for inference on the conditional total policy effect.

\begin{thm}\label{thm: root-n consistency of the ATT(r)}
	Under Assumptions \ref{ass:G3}-\ref{ass:smoothness}, when $h_n\asymp b_n\asymp n^{-1/4-\delta}$, for $\delta\in (0,1/12)$, the estimator of the counterfactual total policy effect has the following linear expansion
	\[
	\begin{split}
		&\quad \sqrt{n}{F}^{cf}(c^{cf})\left(\widehat{ATT}(c^{cf})-ATT(c^{cf})\right)\\
		&= \frac{1}{\alpha \sqrt{n}}\sum_{i:Z_i=1} \frac{\eta_{i1}f^{cf}(R_{i1})}{f_1(R_{i1})}F_{K^*}\left(\frac{R_{i1}-c^{cf}}{b_n}\right) \\
		&+\frac{1}{\alpha \sqrt{n}} \sum_{i:Z_i=1} m_1(R_{i1}+\Delta c) \left\{F_{K,i}-E\left[F_{K,i}|R_{i0}\right] \right\} \frac{f_0(R_{i0}+\Delta c)}{f_0 (R_{i0})}\\
		&+\frac{1}{\alpha \sqrt{n}}  \sum_{i:Z_i=1}  \left\{\int_{r\ge c^{cf}} m_1(r)f_{1|0}(r_1-\Delta c|R_{i0}-\Delta c) dr - E\left[\int_{r\ge c^{cf}} m_1(r)f_{1|0}(r_1-\Delta c|R_{i0}-\Delta c) \right] \right\}\\
		&-\frac{1}{(1-\alpha)\sqrt{n}}\sum_{i:Z_i=0} \frac{f_0(R_{i0})}{\tilde{f}_0(R_{i0})} \varepsilon_i \kappa(R_{i0}) -  \frac{1}{\alpha \sqrt{n}}\sum_{i:Z_i=1}\eta_{i0} \kappa(R_{i0})\\
		&-\frac{1}{\alpha\sqrt{n}}\sum_{i: Z_i=1} \left[\theta(R_{i0})(1-F_{1|0}(c^{cf}-\Delta c|R_{i0}-\Delta c))-E[\theta(R_{i0})(1-F_{1|0}(c^{cf}-\Delta c|R_{i0}-\Delta c)|Z_i=1]\right]\\
		&- \frac{1}{\alpha\sqrt{n}} 
		\sum_{i: Z_i = 1} 
		\pi(R_{i0})
		\left[
		F_K\!\left( \frac{R_{i1} - (c^{cf} - \Delta c)}{b_n} \right)
		- 
		\mathbb{E}\!\left\{
		F_K\!\left( \frac{R_{i1} - (c^{cf} - \Delta c)}{b_n} \right)
		\Bigm| R_{i0}
		\right\}
		\right]+o_p(1),
	\end{split}
	\]
	where $\eta_{i1}:=Y_{i1}-m_1(R_{i1})$, $\varepsilon_i=Y_{i1}-Y_{i0} - E[Y_{i1}-Y_{i0}|R_{i0},Z_i=0]$, $\kappa(r_0)=(1-F_{1|0}(c^{cf}-\Delta c|R_{i0}=r_0-\Delta c))$, $F_{1|0}(a,b)=Pr(R_{i1}\le a| R_{i0}=b,Z_i=1)$, $\tilde{f}_0(r_0)=f(R_{i0}=r_0|Z_i=0)$, $F_K(t)=\int_{t}^\infty K(u) du$ for kernel $K$, and $\pi(x) = 
	\frac{
		\theta(x + \Delta c)\, f_0(x + \Delta c)
	}{
		f_{0}(x)
	}
	$.
\end{thm}
\begin{proof}
	Follows directly from Lemmas \ref{lem: linearization for the ATT integrant} and  \ref{lem: linearization for m1 integration for ATT}.
\end{proof}

We achieve $\sqrt{n}$ convergence for the $\widehat{ATT}(c^{cf})$ even if we use nonparametric density estimation in the process. This is because we use the undersmoothing in the bandwidth. The ratio $ \frac{f_0(r)}{\tilde{f}_0(r)}$ reflects the covariate $R_{i0}$ density shift across the two groups.

\begin{Corollary}\label{cor:att_clt_quad_compact}
	Let the conditions of Theorem \ref{thm: root-n consistency of the ATT(r)} and Assumption \ref{ass: balanced sample} hold, 
	assume a multivariate Lindeberg--Feller central limit theorem for the terms below,
	and suppose $\widehat{F}^{cf}(c^{cf}) \overset{p}{\to} F^{cf}(c^{cf})>0$.
	
	Define, for each observation $i$, the $7\times 1$ random vector
	{\footnotesize
		\[
		\bm{\xi}_i
		=
		\begin{pmatrix}
			\mathbf{1}\{Z_i=1\}
			\dfrac{\eta_{i1} f^{cf}(R_{i1})}{f_1(R_{i1})}
			F_{K^*}\!\left(\dfrac{R_{i1}-c^{cf}}{b_n}\right)
			\\[1.2em]
			\mathbf{1}\{Z_i=1\}
			m_1(R_{i1}+\Delta c)
			\Big\{
			F_{K,i}-\mathbb{E}[F_{K,i}\mid R_{i0}]
			\Big\}
			\dfrac{f_0(R_{i0}+\Delta c)}{f_0(R_{i0})}
			\\[1.2em]
			\mathbf{1}\{Z_i=1\}
			\Bigg[
			\int_{r\ge c^{cf}}
			m_1(r)\,
			f_{1|0}(r-\Delta c \mid R_{i0}-\Delta c)\,dr
			-
			\mathbb{E}\Big[
			\int_{r\ge c^{cf}}
			m_1(r)\,
			f_{1|0}(r-\Delta c \mid R_{i0}-\Delta c)\,dr
			\Big]
			\Bigg]
			\\[1.2em]
			\mathbf{1}\{Z_i=0\}
			\Bigg[
			-\frac{f_0(R_{i0})}{\tilde{f}_0(R_{i0})}
			\,\varepsilon_i \,\kappa(R_{i0})
			\Bigg]
			\\[1.2em]
			\mathbf{1}\{Z_i=1\}
			\Big[
			-\eta_{i0}\,\kappa(R_{i0})
			\Big]
			\\[1.2em]
			\mathbf{1}\{Z_i=1\}
			\Bigg[
			-\Big(
			\theta(R_{i0})
			\big(
			1-F_{1|0}(c^{cf}-\Delta c \mid R_{i0}-\Delta c)
			\big)
			-
			\mathbb{E}\big[
			\theta(R_{i0})
			(1-F_{1|0}(c^{cf}-\Delta c \mid R_{i0}-\Delta c))
			\mid Z_i=1
			\big]
			\Big)
			\Bigg]
			\\[1.2em]
			\mathbf{1}\{Z_i=1\}
			\Bigg[
			-\pi(R_{i0})
			\Big(
			F_K\!\left(
			\dfrac{R_{i1}-(c^{cf}-\Delta c)}{b_n}
			\right)
			-
			\mathbb{E}\!\left\{
			F_K\!\left(
			\dfrac{R_{i1}-(c^{cf}-\Delta c)}{b_n}
			\right)
			\Bigm| R_{i0}
			\right\}
			\Big)
			\Bigg]
		\end{pmatrix}.
		\]}
	
	Let $V=\mathbb{V}(\bm{\xi}_i)$ denote the $7\times 7$ covariance matrix of $\bm{\xi}_i$, and define
	$
	\bm{\alpha}
	=\frac{1}{{F}^{cf}(c^{cf})}	(1/\alpha ,
	1/\alpha ,
	1/\alpha,
	1/(1-\alpha) ,
	1/\alpha ,
	1/\alpha ,
	1/\alpha)'
	$, then the linear expansion in Theorem \ref{thm: root-n consistency of the ATT(r)} can be written as
	\[
	\sqrt{n}
	\big(\widehat{ATT}(c^{cf})-ATT(c^{cf})\big)
	=
	\sqrt{n}\;
	\bm{\alpha}' \,\bar{\bm{\xi}}_n
	+o_p(1),
	\qquad
	\bar{\bm{\xi}}_n
	:=\frac{1}{n}\sum_{i=1}^n \bm{\xi}_i.
	\]
	
	Suppose $\bm{\xi}_i$ satisfies the Lindeberg condition for CLT, then
	\[
	\sqrt{n}\,\big(\widehat{ATT}(c^{cf})-ATT(c^{cf})\big)
	\;\overset{d}{\longrightarrow}\;
	\mathcal{N}\!\left(
	0,\;\bm{\alpha}' V \bm{\alpha}
	\right).
	\]
\end{Corollary}

A consistent estimator of $V$ can be found by plugging in consistent estimators to first estimate a sample $\hat{\bm{\xi}}_i$, and then find the covariance matrix of $\hat{\bm{\xi}}_i$. 

\subsection{Bootstrap Inference}
If we rely on Corollary \ref{cor:att_clt_quad_compact} to do inference on the $ATT$, then we need to estimate $\bm{\xi}_i$, which can be complicated to implement as we need to estimate multiple nonparametric objects. We now propose the bootstrap estimator which can be easier to implement at the cost of slightly more computational power.

\begin{prop}\label{prop: bootstrap validity}
	Let $b=1,2,...,B$ be the empirical bootstrap sample with replacement, and denote $(R_{i1}^{*,b},R_{i0}^{*,b},Y_{i1}^{*,b},Y_{i0}^{*,b},Z_i^{*,b})$ the bootstrap data for sample $b$.  Consider the balanced-group bootstrap, i.e., generate the bootstrap sample $b$ with exactly $n_0$ of $Z_i^{*,b}=0$ and $n_1$ of $Z_{i}^{*,b}=1$.  Suppose $\widehat{\bm{\xi}}_i$ is sup-norm consistent for ${\bm{\xi}}_i$ such that $\sup_{i} |\widehat{\bm{\xi}}_i-\bm{\xi}_i|\rightarrow_p 0$. Now let $ATT^{*,b}(c^{cf})$ be the bootstrapped counterfactual estimator of the total policy effects, and let $c^*_{\alpha/2}$ and $c^*_{1-\alpha/2}$ be the empirical $\alpha$ and $1-\alpha/2$ quantile of $(ATT^{*,b}(c^{cf})-\widehat{ATT}(c^{cf}))$, then the confidence interval 
	\[
	\left[\widehat{ATT}(c^{cf})-c^*_{1-\alpha/2},\, \widehat{ATT}(c^{cf})+c^*_{1-\alpha/2} \right]
	\]
	is a valid $\alpha$-confidence interval for the true counterfactual $ATT(c^{cf})$. 
\end{prop}

\section{Simulation Study}\label{section: Simulation}

We simulate two groups of units, indexed by $Z_i \in \{0,1\}$, each observed over two periods $t=0,1$.
For the control group ($Z_i=0$), we observe the undistorted potential running variable 
$R_{it}^0$ at both periods, while for the treated group ($Z_i=1$), we observe 
$R_{i0}^0$ at $t=0$ and the distorted running variable $R_{i1}^1$ at $t=1$ due to policy distortion of the running variable.

\paragraph{Running variable dynamics.}
\begin{align*}
	R_{i0}^0 &\sim \text{Truncated Normal}(50,\,20^2;\,[20,80]), & Z_i=0, \\
	R_{i1}^0 &= R_{i0}^0 + N(1,\,4), & Z_i=0, \\
	R_{i0}^0 &\sim \text{Truncated Normal}(60,\,15^2;\,[20,80]), & Z_i=1, \\
	R_{i1}^1 &= R_{i0}^0 + N(2,\,1), & Z_i=1,
\end{align*}
where the mean shift in $R_{i1}^1$ reflects endogenous distortion when the treatment is potentially available. The simulation setting of $R_{it}$ is consistent with the translation invariant selection Assumption \ref{assumption: selection behavior}.

\paragraph{Potential outcome functions.}
For the untreated potential outcome $Y_{it}^0(R_{it}^0)$, the conditional mean functions are defined as follows.

For the control group ($Z_i=0$):
\begin{align*}
	E[Y_{i0}^0 \mid R_{i0}^0 = r,\, Z_i=0] &= 5\log(r+1) + 15\exp\!\left(-\frac{(r-50)^2}{200}\right),\\
	E[Y_{i1}^0 \mid R_{i1}^0 = r,\, Z_i=0] &= 5\log(r+1) + 15\exp\!\left(-\frac{(r-50)^2}{200}\right)+2\left(3\cos^2\!\left(\frac{\pi}{60}(r-50)\right) + 1 \right).
\end{align*}

For the treated group ($Z_i=1$), the untreated potential outcomes in the first period and its counterfactual path in the second period satisfy:
\begin{align*}
	E[Y_{i0}^0 \mid R_{i0}^0 = r,\, Z_i=1] &= 6\log(r+1) + 11\exp\!\left(-\frac{(r-50)^2}{200}\right),\\
	E[Y_{i1}^0 \mid R_{i1}^0 = r,\, Z_i=1] &= 6\log(r+1) + 11\exp\!\left(-\frac{(r-50)^2}{200}\right)
	+2\left(3\cos^2\!\left(\frac{\pi}{60}(r-50)\right) + 1 \right).
\end{align*}
The conditional trend function for the untreated potential outcome is given by $2\left(3\cos^2\!\left(\frac{\pi}{60}(r-50)\right) + 1 \right)$. The unobserved potential outcome $Y_{i1}^0$ for the $Z_i=1$ group is explicitly simulated so that we can take the difference $Y_{i1}^1-Y_{i1}^0|Z_i=1$ to calculate the total policy effect for individual $i$. 

The realized post-treatment outcomes for the $Z_i=1$ group in period 1 are
\begin{align*}
	E[Y_{i1}^1 \mid R_{i1}^1 = r,\, Z_i=1] &= 6.5\log(r+1) + 13\exp\!\left(-\frac{(r-50)^2}{200}\right)
	+ 7\cos^2\!\left(\frac{\pi}{60}(r-60)\right) + 3.
\end{align*}

\paragraph{Cutoff and counterfactual cutoff choices.} For the $Z_i=1$ group, we use $c=60$ as the actual policy cutoff so that for $R_{it}\ge 60$, we observe $Y_{i1}^1$ as $Y_{i1}$ for the $Z_i=1$ group. The counterfactual cutoff is chosen at $c^{cf}=63$.

\subsection{Performance}

\begin{table}[H]
	\centering
	\begin{tabular}{|c|cc|cc|cc|}
		\cline{1-7}
		\multicolumn{1}{|l|}{}                & \multicolumn{2}{c|}{$n^{-0.28}\sigma_R$}                             & \multicolumn{2}{c|}{$n^{-0.30}\sigma_R$}                             & \multicolumn{2}{c|}{$n^{-0.32}\sigma_R$}                             \\ \cline{2-7} 
		& \multicolumn{1}{c|}{\textit{original}}      & \textit{bias-reduced}  & \multicolumn{1}{c|}{\textit{original}}      & \textit{bias-reduced}  & \multicolumn{1}{c|}{\textit{original}}      & \textit{bias-reduced}  \\\cline{1-7}
		\multirow{2}{*}{$\frac{1}{\sqrt{2}}$} & \multicolumn{1}{c|}{\multirow{2}{*}{5.980}} & \multirow{3}{*}{6.240} & \multicolumn{1}{c|}{\multirow{2}{*}{6.041}} & \multirow{3}{*}{6.235} & \multicolumn{1}{c|}{\multirow{2}{*}{6.087}} & \multirow{3}{*}{6.233} \\
		& \multicolumn{1}{c|}{}                       &                        & \multicolumn{1}{c|}{}                       &                        & \multicolumn{1}{c|}{}                       &                        \\ \cline{1-2} \cline{4-4} \cline{6-6}
		\multirow{2}{*}{$1$}                  & \multicolumn{1}{c|}{\multirow{2}{*}{5.720}} &                        & \multicolumn{1}{c|}{\multirow{2}{*}{5.847}} &                        & \multicolumn{1}{c|}{\multirow{2}{*}{5.941}} &                        \\ \cline{3-3} \cline{5-5} \cline{7-7} 
		& \multicolumn{1}{c|}{}                       & \multirow{3}{*}{6.240} & \multicolumn{1}{c|}{}                       & \multirow{3}{*}{6.245} & \multicolumn{1}{c|}{}                       & \multirow{3}{*}{6.241} \\ \cline{1-2} \cline{4-4} \cline{6-6}
		\multirow{2}{*}{$\sqrt{2}$}           & \multicolumn{1}{c|}{\multirow{2}{*}{5.200}} &                        & \multicolumn{1}{c|}{\multirow{2}{*}{5.449}} &                        & \multicolumn{1}{c|}{\multirow{2}{*}{5.641}} &                        \\
		& \multicolumn{1}{c|}{}                       &                        & \multicolumn{1}{c|}{}                       &                        & \multicolumn{1}{c|}{}                       &                        \\ \cline{1-7}
	\end{tabular}
	\caption{Bias of Original Estimator and Bias Reduction for $n=1000$.}
	\label{table: bias for n=1000}
\end{table}

\begin{table}[H]
	\centering
	\begin{tabular}{|c|cc|cc|cc|}
		\cline{1-7}
		\multicolumn{1}{|l|}{}                & \multicolumn{2}{c|}{$n^{-0.28}\sigma_R$}                             & \multicolumn{2}{c|}{$n^{-0.30}\sigma_R$}                             & \multicolumn{2}{c|}{$n^{-0.32}\sigma_R$}                             \\ \cline{2-7} 
		& \multicolumn{1}{c|}{\textit{original}}      & \textit{bias-reduced}  & \multicolumn{1}{c|}{\textit{original}}      & \textit{bias-reduced}  & \multicolumn{1}{c|}{\textit{original}}      & \textit{bias-reduced}  \\\cline{1-7}
		\multirow{2}{*}{$\frac{1}{\sqrt{2}}$} 
		& \multicolumn{1}{c|}{\multirow{2}{*}{6.053}} & \multirow{3}{*}{6.209} 
		& \multicolumn{1}{c|}{\multirow{2}{*}{6.093}} & \multirow{3}{*}{6.208} 
		& \multicolumn{1}{c|}{\multirow{2}{*}{6.124}} & \multirow{3}{*}{6.208} \\
		& \multicolumn{1}{c|}{}                       &                        
		& \multicolumn{1}{c|}{}                       &                        
		& \multicolumn{1}{c|}{}                       &                        \\ \cline{1-2} \cline{4-4} \cline{6-6}
		\multirow{2}{*}{$1$}                  
		& \multicolumn{1}{c|}{\multirow{2}{*}{5.896}} &                        
		& \multicolumn{1}{c|}{\multirow{2}{*}{5.979}} &                        
		& \multicolumn{1}{c|}{\multirow{2}{*}{6.039}} &                        \\ \cline{3-3} \cline{5-5} \cline{7-7} 
		& \multicolumn{1}{c|}{}                       & \multirow{3}{*}{6.220} 
		& \multicolumn{1}{c|}{}                       & \multirow{3}{*}{6.214} 
		& \multicolumn{1}{c|}{}                       & \multirow{3}{*}{6.210} \\ \cline{1-2} \cline{4-4} \cline{6-6}
		\multirow{2}{*}{$\sqrt{2}$}           
		& \multicolumn{1}{c|}{\multirow{2}{*}{5.572}} &                        
		& \multicolumn{1}{c|}{\multirow{2}{*}{5.744}} &                        
		& \multicolumn{1}{c|}{\multirow{2}{*}{5.868}} &                        \\
		& \multicolumn{1}{c|}{}                       &                        
		& \multicolumn{1}{c|}{}                       &                        
		& \multicolumn{1}{c|}{}                       &                        \\ \cline{1-7}
	\end{tabular}
	\caption{Bias of Original Estimator and Bias Reduction for $n=2000$.}
	\label{table: bias for n=2000}
\end{table}

\paragraph{Performance of the Original Estimator.}
Each cell under the \textit{original} columns in Tables~\ref{table: bias for n=1000} and~\ref{table: bias for n=2000} corresponds to the mean of ATT estimator obtained using a conventional estimator with bandwidth
\[
h = \text{(row constant)} \times n^{-x}\sigma_R,
\]
where the exponent $x \in \{0.28, 0.30, 0.32\}$ is indicated by the column header, and the row constant is given by the leftmost row label ($1/\sqrt{2}$, $1$, or $\sqrt{2}$).
For instance, in the $n=1000$ table, the value $5.980$ under the $n^{-0.28}\sigma_R$ column and the $1/\sqrt{2}$ row represents the estimate obtained using the bandwidth $h = \frac{1}{\sqrt{2}} \times 1000^{-0.28}\sigma_R$. The true total policy effect is  $\text{ATT}^{true} = 6.195$. In both $n=1000$ and $n=2000$ cases, all \textit{original} estimates lie below the true value, indicating negative bias. The bias of the estimators reduces with the choice of bandwidth, but is still significant when we use the smallest bandwidth $n^{-0.32}\sigma_R/\sqrt{2}$. This indicates that the asymptotic bias $\sqrt{nh^2}$ is not fully washed away in the finite sample.

\paragraph{Bias-Reduction Method.}
The \textit{bias-reduced} columns implement a two-scale bias-correction scheme of the form
\[
\widehat{ATT}_{\text{BR}}(h)
= 2\widehat{ATT}(h) - \widehat{ATT}(\sqrt{2}\,h),
\]
where $\widehat{ATT}(h)$ denotes the original estimator using bandwidth $h$, and we suppress the dependence on the counterfactual cutoff $c^{cf}$ in the following notation.
This linear combination removes the leading-order bias term $O(\sqrt{n}h^2)$ in the estimator and leaves us with the higher-order bias $O(\sqrt{n}h^4)$ \citep{jones1995simple}.

Empirically, the bias-reduced values in both tables are substantially closer to the true ATT ($6.195$) than the original estimates.
For example, in the $n=1000$ table under $n^{-0.28}\sigma_R$ and $1/\sqrt{2}$, the original estimate is $5.980$ (bias $\approx -0.215$) while the bias-reduced estimate is $6.240$ (bias $\approx 0.045$).
The same improvement holds across other bandwidths and constants, and the performance further improves for $n=2000$.
Overall, the results confirm that the two-scale correction effectively mitigates the leading bias term, producing estimates with smaller absolute bias across different bandwidths and sample sizes.

As a result, we recommend using the $\widehat{ATT}_{\text{BR}}(h)$ in relatively small samples to avoid bias in the estimator.

\paragraph{Bootstrap Adjustment.} Since we use the bias-reduced estimator $\widehat{ATT}_{\text{BR}}(h)$, we correspondingly change the bootstrap procedure and let 
\[
\widehat{ATT}^{*,b}_{BR}(h)=2\widehat{ATT}^{*,b}_{BR}(h)-\widehat{ATT}^{*,b}_{BR}(\sqrt{2}h),
\]
where $\widehat{ATT}^{*,b}_{BR}(h)$ is estimated from the $b$-th bootstrap sample, and construct the confidence interval using the $\alpha/2$ and $1-\alpha/2$ quantile of $2\widehat{ATT}_{BR}(h)-\widehat{ATT}^{*,b}_{BR}(h)$ as the confidence interval. Alternatively, we can use the t-bootstrap to calculate the variance of the bootstrap estimator $\sigma_{ATT,BR}^*$ and construct the confidence interval as $[\widehat{ATT}_{BR}(h)\pm 1.96\sigma_{ATT,BR}^*]$.

The validity of the bootstrap is ensured by looking at the linear expansion in Theorem \ref{thm: root-n consistency of the ATT(r)} and Proposition \ref{prop: bootstrap validity}. Let $\xi_i(h)$ denote the leading influence term in Theorem \ref{thm: root-n consistency of the ATT(r)} under bandwidth of $h$, then, fixing and suppressing the counterfactual cutoff notation $c^{cf}$, we have 
\[
\begin{split}
	\sqrt{n}\left(\widehat{ATT}_{BR}(h)-ATT^{true}\right) &=\frac{1}{n}\sum_{i} 2\xi_i(h)-\xi(\sqrt{2}h),\\
	\sqrt{n}\left({ATT}^{*,b}_{BR}(h)-\widehat{ATT}_{BR}(h)\right) &=\frac{1}{n}\sum_{i} \left[ 2\hat{\xi}_i(h)-\hat{\xi}_i(\sqrt{2}h)\right](W_{i}^{*,b}-1),\\
\end{split}
\]
where $W_{i}^{*,b}$ is the resampling weights for observation $i$. The validity of the bootstrap is established as $\left[ 2\hat{\xi}_i(h)-\hat{\xi}_i(\sqrt{2}h)\right]$ mimics the variance structure of $2\xi_i(h)-\xi(\sqrt{2}h)$. 

The bootstrap method here does not require us to take into account the additional estimation variation in correcting the bias term. This is crucially because we use the undersmoothing and the additional variance introduced by correcting for the bias is of order $\sqrt{n}h^2\rightarrow 0$, which is washed away in the limit and does not matter when we use the bootstrap method. 

\begin{table}[H]
	\centering
	\caption{Coverage probabilities of the 95\% confidence intervals, n=2000 }
	\begin{tabular}{lccc}
		\cline{1-4}
		bandwidth& $ n^{-0.28}\sigma_R/\sqrt{2}$ 
		& $ n^{-0.30}\sigma_R/\sqrt{2}$ 
		& $ n^{-0.32}\sigma_R/\sqrt{2}$ \\
		\cline{1-4}
		Empirical bootstrap & 0.925 & 0.935 & 0.935 \\
		$t$-Bootstrap       & 0.960 & 0.925 & 0.930 \\
		Oracle Bias removed  Empirical Bootstrap  &0.940 &0.945 &0.940 \\
		\cline{1-4}
	\end{tabular}
	\label{table: bootstrap validity, n=2000}
\end{table}
As Table \ref{table: bootstrap validity, n=2000} shows, both empirical and t-bootstrap methods have some undercoverage problem except for the $t$-bootstrap in the $n^{-0.28}\sigma_R/\sqrt{2}$ case. Undercoverage issue is pervasive in the nonparametric inference problem \citep{calonico2014robust,armstrong2020simple}. This is because, even if we use the bias-reduced estimator, we still have the stochastic error of order $O_p(h^2+b_n^2+\log n/(nh_n\sqrt{b_n}))$, which influences the coverage probability in the finite sample size. To see how this bias can influence the coverage, we present the ``Oracle Bias removed  Empirical Bootstrap" in Table \ref{table: bootstrap validity, n=2000}, in which we oracularly remove the bias in the original estimator. Further removing the bias makes the coverage probability even closer to the nominal coverage probability of 95\%. 

\section{Empirical Illustration: Fiscal Rules in Italian Municipalities} \label{section: Empirical}

To apply our framework, we consider the empirical setting from \cite{grembi2016fiscal}. Their study examines the effect of fiscal rules imposed by the Italian central government on local municipalities. Beginning in 1999, the Domestic Stability Pact (DSP) required all municipalities to limit the growth of their fiscal gap—defined as the difference between total expenditure net of debt service and total revenue net of transfers. In 2001, the Italian government relaxed these rules for municipalities with populations below $5{,}000$, while the rules remained binding for those above this threshold. The key policy variable thus depends on whether a municipality’s population size lies above or below this administratively determined cutoff. 

In our notation, we denote the running variable by $R_{it}$, representing the population size of municipality $i$ at time $t$. The treatment indicator is defined as
\[
D_{it} = \mathbbm{1}(R_{it} \le 5000)\times \mathbbm{1}(t \geq 2001),
\]
which equals one if the municipality is subject to the fiscal rule in year $t$ and zero otherwise. The outcome variable, $Y_{it}$, represents the fiscal balance of the municipality, measured by the per-capita deficit (total expenditure minus total revenue) as in \cite{grembi2016fiscal}. The panel spans $t = 1997, \dots, 2004$, covering both pre- and post-policy periods. For our purposes, we define the pre-policy period $(t=1999,2000)$ as the control group with $Z_i=0$, which will be used to establish the conditional trend, and the early post-policy period $(t=2000,2001)$ as the policy-influenced group with $Z_i=1$. This two-period framing allows us to align the \cite{grembi2016fiscal} policy variation with the structure of our baseline model in Section \ref{section: model}.

\cite{grembi2016fiscal} also notes that a discontinuity in mayoral wages occurs at the same $5{,}000$ population threshold, which could contaminate identification. However, in our setup, this concern is mitigated by the conditional parallel trend assumption \ref{assumption: conditional parallel trend assumption in running variable}: because the mayor’s wage is a deterministic function of the running variable $R_{it}$, its effect is fully absorbed once we condition on $R_{it}$. Thus, any wage-related differences across municipalities do not bias the identification of the fiscal-rule effect.

Finally, the population variable $R_{it}$ may itself be distorted by policy incentives. While mayors cannot directly choose the population, residents may migrate toward municipalities with more relaxed fiscal constraints due to enhanced local services. In addition, fiscal rules can influence demographic changes through birth and death rates over time. Consequently, policy effects may manifest not only near the cutoff but also through broader shifts in the population distribution. Therefore, when evaluating the total policy effect, it is crucial to take into account the indirect effect via the population variation.

\subsection{Results}
The results of the estimation are shown in Table \ref{tab:cf_estimates}. Recall that the policy is imposed targeting to constrain local municipalities from fiscal expansion, so a negative number indicates that the policy is successful in reducing the deficit.

\begin{table}[htbp]
	\centering
	\begin{tabular}{ccc}
		\cline{1-3}
		Counterfactual value & Point estimate & Confidence interval \\
		\cline{1-3}
		5000 & -7.37 & [-17.06, 5.09] \\
		4900 & -7.21 & [-17.81, 5.91] \\
		4800 & -5.83 & [-16.98, 7.36] \\
		4700 & -3.78 & [-16.33, 14.06] \\
		4600 & 0.45 & [-15.56, 27.04] \\
		4500 & 7.18 & [-14.89, 42.87] \\
		4400 & 16.13 & [-8.22, 70.09] \\
		4300 & 27.45 & [-8.29, 102.01] \\
		4200 & 40.97 & [-1.81, 136.56] \\
		4100 & 59.36 & [-8.65, 195.12] \\
		4000 & 90.08 & [11.15, 294.59] \\
		\cline{1-3}
	\end{tabular}
	\caption{Point estimates and confidence intervals at counterfactual values}
	\label{tab:cf_estimates}
\end{table}

The results show that the slash-down of the deficit is decreasing as we decrease the counterfactual cutoff to the 4,000 population. After the cutoff is moved to 4,600, the total effect of policy will be positive, which shows that the deficit cutdowns are mostly from the large population towns. The result is quite different from the \cite{grembi2016fiscal} as they show that towns around the policy cutoff have increased deficits due to the policy relaxation. Of course, the objects identified in  \cite{grembi2016fiscal} are the local direct effect of treatment on the deficit while our results also incorporate the change in the population due to the potential of treatment. The confidence intervals are wide and cover zero most of the time. This is probably because of the high variance in the response of municipal fiscal conditions to the policy and the relatively small sample size.

\appendix

\section{Proofs}

\subsection{Proofs of Identification Results}
\subsubsection{Proof of Theorem \ref{theorem: identify conditional total treatment effect}}
\begin{proof}
	By definition of the total local treatment effect:
	\[
	T(c^{cf},r)\equiv E[Y_{i1}^1(R^1_{i1}(c^{cf}))-Y_{i1}^0(R_{i1}^0)|Z_i=1,R^{1}_{i1}(c^{cf})=r].
	\]
	The first part, by the potential outcome framework \eqref{eq: potential outcome for Y}, can be written as 
	\[
	\begin{split}
		E[Y_{i1}^1(R^{1}_{i1}(c^{cf}))|Z_i=1,R^{1}_{i1}(c^{cf})=r]&= 	E[Y_{i1}^1(R_{i1})|Z_i=1,R_{i1}=r]\\
		&=_{i.d.}E[Y_{i1}|Z_i=1,R_{i1}=r]
	\end{split}
	\]
	because $Y^1_{i1}(R_{i1})=Y_{i1}$ for the $Z_i=1$ group.
	
	For the second term, we have 
	\[
	\begin{split}
		&\quad E[Y_{i1}^0(R_{i1}^0)|Z_i=1,R^{1}_{i1}(c^{cf})=r]\\
		&=_{(a)}\int_{r_0} E[Y_{i1}^0(R_{i1}^0)|Z_i=1,R^{1}_{i1}(c^{cf})=r, R_{i0}^0=r_0] dF(R_{i0}^0=r_0 | Z_i=1, R^{1}_{i1}(c^{cf})=r)\\
		&=_{(b)}\int_{r_0} E[Y_{i1}^0(R_{i1}^0)|Z_i=1,R_{i0}^0=r_0] dF(R_{i0}^0=r_0 | Z_i=1, R^{1}_{i1}(c^{cf})=r)\\
		&=_{(c)} \int_{r_0} E[Y_{i1}^0(R_{i1}^0)|Z_i=1,R_{i0}^0=r_0]  dF(R_{i0}^0=r_0|R^{1}_{i1}(c^{cf})=r)
	\end{split}\]
	where $(a)$ follows from the law of iterative expectation, $(b)$ follows from the conditional exogeneity of the potential outcome (Assumption \ref{assumption: conditional exogeneity of the potential outcome }), $(c)$ follows from the Lemma \ref{lem: counterfactual distribution of R_{it} in Z_i=1}.
\end{proof}

\subsection{Useful Lemmas for Deriving Asymptotic Result}

\begin{lem}[Second-order expansion for $1/\hat g(x)$] \label{lemma: expansion for 1/g(x)}
	Let $g:\mathcal X\to\mathbb R$ and an estimator $\hat g:\mathcal X\to\mathbb R$ satisfy the linearization
	\[
	{\ \hat g(x)\;=\;g(x)\;+\;\frac{1}{n}\sum_{i=1}^n \phi_i(x)\;+\;\mathrm{bias}_n(x)\ }
	\]
	for $x\in\mathcal X$. Suppose that $g(x)$ is uniformly bounded away from zero, and for $r_n\rightarrow 0 $ and $h_n\rightarrow 0$,
	\[
	\sup_x \left| \frac{1}{n}\sum_{i=1}^n \phi_i(x) \right| = o_p( r_n),\quad \sup_x \left| bias(x) \right| = O( h^2_n),
	\]
	then
	\[
	\sup_x \left|\frac{1}{\hat g(x)}
	\;-\;\frac{1}{g(x)}
	\;-\;\frac{1}{g(x)^2}\Big(\frac{1}{n}\sum_{i=1}^n \phi_i(x)+\mathrm{bias}_n(x)\Big)
	\;\right|= O_p( r_n^2+h_n^4 ).
	\]
\end{lem}

\begin{proof}
	Directly by the taylor exapnsion \[
	\frac{1}{\hat g(x)}
	\;=\;\frac{1}{g(x)}
	\;-\;\frac{1}{g(x)^2}\Big(\frac{1}{n}\sum_{i=1}^n \phi_i(x)+\mathrm{bias}_n(x)\Big)
	\;+\;\frac{1}{\widetilde{g(x)}^3}\Big(\frac{1}{n}\sum_{i=1}^n \phi_i(x)+\mathrm{bias}_n(x)\Big)^{\!2},
	\]
	where $\widetilde{g(x)}$ lies between $g(x)$ and $\hat{g}(x)$, and $\widetilde{g(x)}$ with probability approaching one bounded away from zero. The remainder term $\frac{1}{\widetilde{g(x)}^3}\Big(\frac{1}{n}\sum_{i=1}^n \phi_i(x)+\mathrm{bias}_n(x)\Big)^{\!2}$ is of order $O_p( (r_n+h_n^2)^2)$, wich leads to the result in the lemma.
\end{proof}

\begin{lem}[Uniform linearization and product expansion for $\hat g_1(x)$ and $\hat g_2(y)$]\label{lem:prod-linearization}
	Let $g_1(x),g_2(y)$ be uniformly bounded real-valued targets and let $\hat g_1(x),\hat g_2(y)$ be estimators based on (possibly different) samples
	of sizes $n_1,n_2$ such that $n_1=\alpha_1 n$, and $n_2=\alpha_2 n$. Suppose we have the linear expansions for $d=1,2$:
	\begin{equation}\label{eq:g1-lin}
		\sup_{x}\left|\hat g_d(x)-g_d(x) -\frac{1}{n_1}\sum_{i=1}^{n_d}\phi_{di}\right| = O_p(r_{n,d}).
	\end{equation}
	
	Moreover, the influence function term has the stochasitc order $\sup_x \left| \frac{1}{n_d}\sum_{i=1}^{n_d}\phi_{di} \right|=o_p(\tilde{r}_{n,d})$ for $d=1,2$. Then the product admits the uniform linearization:
	
		\begin{equation}\label{eq:prod-main2}
		{\;
			\sup_x\left|\hat g_1\hat g_2 - g_1 g_2
			-\left[
			g_2  \frac{1}{n_1}\sum_{i=1}^{n_1}\phi_{1i}
			\;+\;
			g_1 \frac{1}{n_2}\sum_{i=1}^{n_2}\phi_{2i} +\left(\frac{1}{n_1}\sum_{i=1}^{n_1}\phi_{1i}\right)\left(\frac{1}{n_2}\sum_{i=1}^{n_2}\phi_{2i}\right) \right]\right|
			\;=\;O_p( r_{n,1}+r_{n,2}).
			\;}
	\end{equation}
	If we only care about the leading terms, we can also write
	\begin{equation}\label{eq:prod-main}
		{\;
			\sup_x\left|\hat g_1\hat g_2 - g_1 g_2
			-\left[
			g_2  \frac{1}{n_1}\sum_{i=1}^{n_1}\phi_{1i}
			\;+\;
			g_1 \frac{1}{n_2}\sum_{i=1}^{n_2}\phi_{2i}\right]\right|
			\;=\;O_p(\tilde{r}_{n,1}\tilde{r}_{n,2}+ r_{n,1}+r_{n,2}).
			\;}
	\end{equation}
\end{lem}
\begin{proof}
	Equation \eqref{eq:prod-main} holds by directly write out the $\hat{g}_1\hat{g}_2$.
\end{proof}

\begin{lem}[Uniform linearization of $\widehat{\theta}(r_0)$]\label{lem:theta-unif}
	Let $\widehat{\theta}(r_0):=\widehat{g}(r_0)+\widehat{m}_0(r_0)$ with
	$\varepsilon_i:=\Delta Y_i-g(R_{i0})$ (on $\{Z_i=0\}$), $\eta_i:=Y_{i0}-m_0(R_{i0})$ (on $\{Z_i=1\}$).
	
	Then, 
	\[
	\begin{split}
		&\quad \sup_{r_0}\Bigg|\widehat{\theta}(r_0)-\theta(r_0)
		-\left[
		\frac{1}{n_0 h_n \tilde{f}_0(r_0)}\!\sum_{i:Z_i=0}\!K\!\Big(\frac{R_{i0}-r_0}{h_n}\Big)\varepsilon_i
		+
		\frac{1}{n_1 h_n f_0(r_0)}\!\sum_{i:Z_i=1}\!K\!\Big(\frac{R_{i0}-r_0}{h_n}\Big)\eta_{i0} \right]\Bigg|\\
		&= O_p\left((h_n)^2 + {\frac{\log n}{n h_n}}  \right),
	\end{split}
	\]
	where $\theta(r_0):=g(r_0)+m_0(r_0)$, and $\tilde{f}_0(r_0)$ is the density of $R_{i0}$ conditional on $Z_i=0$ group.
\end{lem}
\begin{proof}
	The linearization of the Nadaraya-Watson Estimator for conditional moment Estimation is standard and can be find in the textbook, see Section 3.3.1 of \cite{ullah1999nonparametric}. By Assumption \ref{ass: balanced sample}, we replace the stochastic order related to $n_1$ and $n_0$ by $n$.
	 The stochastic remainder $O_p\left((h_n)^2 + {\frac{\log n}{n h_n}}  \right)$ is slightly different from the stochastic order in \cite{ullah1999nonparametric} with an additional $\log n$ term because we need $r_0$-uniform linearization, which requires an additional $\log n$ term. 
\end{proof}

\begin{lem}[Linearization of the plug-in counterfactual marginal density]\label{lem:lin_fcf}
	Let the counterfactual density of $R_{i1}^1(c^{cf})$ be
	\[
	f^{cf}(r_1):=\int f_{1\mid 0}(r_1-\Delta c\mid r_0-\Delta c)\,f_0(r_0)\,dr_0.
	\]
	Then, uniformly in $r_1$ over compact sets,
	\begin{align*}
		\widehat f^{cf}(r_1)-f^{cf}(r_1)
		&=\frac{1}{n_1}\sum_{i:Z_i=1}\mu_i(r_1)
		\;+\;O(h_n^2+b_n^2)
		\;+\;o_p\!\Big((n_1 b_n h_n/{\log n_1})^{-1}+(n_1 h_n/{\log n_1})^{-1}\Big),
	\end{align*}
	where
	\[
	\begin{split}
		\mu_i(r_1)
		&=\frac{\Big[K_{b_n}(R_{i1}-(r_1-\Delta c))
			-E\!\left[K_{b_n}(R_{i1}-(r_1-\Delta c))\mid R_{i0}\right]\Big]\,f_0(R_{i0}+\Delta c)}{b_n f_{0}(R_{i0})}\\
		&\quad+\Big[f_{1\mid 0}(r_1-\Delta c\mid R_{i0}-\Delta c)
		- \int f_0(r_0)\,f_{1\mid 0}(r_1-\Delta c\mid r_0-\Delta c)\,dr_0\Big].
	\end{split}
	\]
\end{lem}

\begin{proof}
	By standard linearization for the marginal density $\widehat f_0(r_0)$ and the conditional density 
	$\widehat f_{1\mid 0}(r_1-\Delta c\mid r_0-\Delta c)$, we have that uniformly over the compact support of $r_1$,
	\begin{align*}
		\widehat f^{cf}(r_1)-f^{cf}(r_1)
		&=\frac{1}{n_1}\sum_{i:Z_i=1}\zeta_i(r_1)
		\;+\;O(h_n^2+b_n^2)
		\;+\;o_p\!\Big((n_1 b_n h_n/{\log n_1})^{-1}+(n_1 h_n/{\log n_1})^{-1}\Big),
	\end{align*}
	where the influence function is
	\[
	\begin{split}
		\zeta_i(r_1)
		&:=\underbrace{\int 
			\frac{K_{h_n}(R_{i0}-(r_0-\Delta c))}{h_nb_nf_{R_0\mid Z=1}(r_0-\Delta c)}
			\Big[K_{b_n}(R_{i1}-(r_1-\Delta c))
			-f_{1\mid 0}(r_1-\Delta c\mid r_0-\Delta c)\Big]f_0(r_0)\,dr_0}_{\text{conditional density component}}\\
		&\quad+\underbrace{\int 
			\frac{\big[K_{h_n}(R_{i0}-r_0)-f_0(r_0)\big]}{h_n}
			\,f_{1\mid 0}(r_1-\Delta c\mid r_0-\Delta c)\,dr_0}_{\text{marginal }f_0\text{ component}}.
	\end{split}
	\]
	
	For the first integral, let $u_1=(R_{i0}-(r_0-\Delta c))/h_n$ and for the second, let $u_2=(R_{i0}-r_0)/h_n$. 
	After the change of variables and a second-order Taylor expansion in $u_1h_n$ and $u_2h_n$, we get
	{\footnotesize
		\[
		\begin{split}
			\zeta_i(r_1)
			&=\int K(u_1)
			\frac{f_0(R_{i0}+\Delta c-u_1h_n)}{b_nf_{R_0\mid Z=1}(R_{i0}-u_1h_n)}
			\Big[K_{b_n}(R_{i1}-(r_1-\Delta c))
			-f_{1\mid 0}(r_1-\Delta c\mid R_{i0}-u_1h_n)\Big]\,du_1\\
			&\quad+\int K(u_2)
			f_{1\mid 0}(r_1-\Delta c\mid R_{i0}-\Delta c-u_2h_n)\,du_2
			-\int f_0(r_0)f_{1\mid 0}(r_1-\Delta c\mid r_0-\Delta c)\,dr_0\\
			&=_{(i)}\frac{\big[K_{b_n}(R_{i1}-(r_1-\Delta c))
				-f_{1\mid 0}(r_1-\Delta c\mid R_{i0})\big]\,f_0(R_{i0}+\Delta c)}{ b_nf_{R_0\mid Z=1}(R_{i0})}\\
			&\quad+\Big[f_{1\mid 0}(r_1-\Delta c\mid R_{i0}-\Delta c)
			- \int f_0(r_0)f_{1\mid 0}(r_1-\Delta c\mid r_0-\Delta c)\,dr_0\Big]
			+\tilde O(h_n^2)+\tilde o_p(h_n^2)\\
			&=\frac{\big[K_{b_n}(R_{i1}-(r_1-\Delta c))-\tilde{\mathcal E}\big]f_0(R_{i0}+\Delta c)}{b_nf_{R_0\mid Z=1}(R_{i0})}
			+\Big[f_{1\mid 0}(r_1-\Delta c\mid R_{i0}-\Delta c)
			- \int f_0(r_0)f_{1\mid 0}(r_1-\Delta c\mid r_0-\Delta c)\,dr_0\Big]\\
			&\quad+\tilde O(h_n^2+b_n^2)+\tilde o_p(h_n^2+b_n^2),
		\end{split}
		\]
	}
	where
	\[
	\tilde{\mathcal E}:=
	E\!\left[K_{b_n}(R_{i1}-(r_1-\Delta c))\mid R_{i0}\right].
	\]
	In step $(i)$ we use the second-order Taylor expansion in the local drift terms $u_1h_n$ and $u_2h_n$, and in the last line we note that
	\(\tilde{\mathcal E}-f_{1\mid 0}(r_1-\Delta c\mid R_{i0}-\Delta c)
	=\tilde O(b_n^2)+\tilde o_p(b_n^2)\).
	This yields the stated form of $\mu_i(r_1)$ and completes the proof.
\end{proof}

\begin{lem}[Uniform linearization of $\hat f_{1|0}(r-\Delta c\mid r_0-\Delta c)\,\hat f_0(r_0)/\hat f^{cf}(r)$]\label{lem:ratio-unif-linearization}
	Set $y:=r-\Delta c$, $x:=r_0-\Delta c$.
	Define
	\[
	\hat\phi(r,r_0):=\hat f_{1|0}(y\mid x)\,\frac{\hat f_0(r_0)}{\hat f^{cf}(r)},\qquad
	\phi(r,r_0):= f_{1|0}(y\mid x)\,\frac{ f_0(r_0)}{ f^{cf}(r)}.
	\]
	Define the influence terms
	\[
	\Psi_{1|0,i}(y,x)
	:= 
	\frac{K\!\left(\dfrac{R_{i0}-x}{h_n}\right)}{h_n\,f_{0}(x)}
	\left[
	\frac{K\!\left(\dfrac{R_{i1}-y}{b_n}\right)}{b_n}
	- f_{1|0}(y\mid x)
	\right], 
	\]
	\[
	\Xi^{(0)}_{i}(r_0)
	:=
	\frac{ K\!\big(\frac{R_{i0}-r_0}{h_n}\big) }{h_n}
	-\mathbb E\!\left[\frac{ K\!\big(\frac{R_{i0}-r_0}{h_n}\big) }{h_n}\right],
	\]
	
	Then, uniformly over $(r,r_0)\in\mathcal I\times\mathcal I_0$,
	\begin{equation}
		{
			\begin{aligned}
				\sup_{r,r_0}\Bigg|	\hat\phi(r,r_0)-\phi(r,r_0)
				&-\bigg[
				\frac{f_0(r_0)}{f^{cf}(r)}\cdot \frac{1}{n_1}\sum_{i:Z_i=1}\Psi_{1|0,i}(y,x)
				\;+\;
				\frac{f_{1|0}(y\mid x)}{f^{cf}(r)}\cdot \frac{1}{n_1}\sum_{i:Z_i=1}\Xi^{(0)}_{i}(r_0)
				\\
				&\quad
				-\frac{f_{1|0}(y\mid x)f_0(r_0)}{f^{cf}(r)^2}\cdot \frac{1}{n_1}\sum_{i:Z_i=1}\mu_{i}(r)\bigg]\Bigg|
				\;\\
				&=O_p\left((h_n)^2+ (b_n)^2 + \frac{\log n_1}{ h_n n_1 \sqrt{b_n} }\right).
		\end{aligned}}
	\end{equation}
\end{lem}

\begin{proof}
	We first note that we have the uniform linear expansion
	\[
	\sup_{y,x}\left|\hat f_{1|0}(y\mid x)- f_{1|0}(y\mid x)-  \frac{1}{n_1}\sum_{i:Z_i=1}\Psi_{1|0,i}(y,x)\right|= O_p\left((h_n)^2+ (b_n)^2 + \frac{\log n_1}{ {h_n n_1} }\right).
	\]
	Note that the stochastic order $ \frac{\log n_1}{ h_n n_1} $ is incurred by envolking Lemma \ref{lemma: expansion for 1/g(x)} for the denominator in the conditional expectation estimation. For the other two terms, we have 
	\[
	\sup_{r_0} \left| \hat{f}_0(r_0)-f_0(r_0)-\frac{1}{n_1}\sum_{i:Z_i=1}\Xi^{(0)}_{i}(r_0)\right|=O\left((h_n)^2\right),
	\]
	and Lemma \ref{lemma: expansion for 1/g(x)} delivers:
	\[
	\sup_{r}\left| \frac{1}{\hat{f}^{cf}(r)}- \frac{1}{f^{cf}(r)}+ \frac{1}{(f^{cf}(r))^2}  \frac{1}{n_1}\sum_{i:Z_i=1}\mu_{i}(r) \right|= O_p\left((h_n)^2+ \frac{\log n_1}{h_nn_1}\right).
	\]
	Then, we apply the uniform convergence rate in \cite{gine2002rates} to get
	\[
	\begin{split}
		\sup_{y,x}\left|\frac{1}{n_1}\sum_{i:Z_i=1}\Psi_{1|0,i}(y,x)\right|&=o_p\left(\sqrt{\frac{\log n_1}{ n_1 b_{n}h_n}}\right),\\
		\sup_{r_0} \left| \frac{1}{n_1}\sum_{i:Z_i=1}\Xi^{(0)}_{i}(r_0)\right|&=o_p\left(\sqrt{\frac{\log n_1}{ n_1 h_n}}\right),\\
		\sup_{r}\left| \frac{1}{n_1}\sum_{i:Z_i=1}\mu_{i}(r) \right|&= o_p\left(\sqrt{\frac{\log n_1}{ n_1 b_n}}\right).
	\end{split}
	\]
	Finally, we apply Lemma \ref{lem:prod-linearization} to two times to finish the proof.
\end{proof}

\begin{lem}[Uniform linearization of $\hat\theta(r_0)\,\hat\phi(r,r_0)-\theta(r_0)\,\phi(r,r_0)$]\label{lem:theta-phi-product}
	Let $\hat\theta(r_0):=\hat g(r_0)+\hat m_0(r_0)$ and
	\[
	\hat\phi(r,r_0):=\hat f_{1|0}(r-\Delta c\,\mid\,r_0-\Delta c)\,\frac{\hat f_0(r_0)}{\hat f^{cf}(r)},
	\qquad
	\phi(r,r_0):= f_{1|0}(r-\Delta c\,\mid\,r_0-\Delta c)\,\frac{ f_0(r_0)}{ f^{cf}(r)},
	\]
	with $\theta(r_0):=g(r_0)+m_0(r_0)$. We have, uniformly over
	$(r,r_0)\in I\times I_0$,
	\[
	\sup_{r,r_0}\left| \hat\theta(r_0)\,\hat\phi(r,r_0)-\theta(r_0)\,\phi(r,r_0) - ( \phi(r,r_0)\,\mathsf{S}_\theta(r_0)
	\;+\; \theta(r_0)\,\mathsf{S}_\phi(r,r_0)) \right|= O_p\left((h_n)^2+ (b_n)^2 +\frac{\log n}{ n h_n \sqrt{b_n } }\right).
	\]
	where the leading stochastic terms are
	\[
	\mathsf{S}_\theta(r_0)
	=
	\frac{1}{n_0\,h_n\,\tilde{f}_0(r_0)}
	\sum_{i:Z_i=0}
	K\!\Big(\frac{R_{i0}-r_0}{h_n}\Big)\,\varepsilon_i
	\;+\;
	\frac{1}{n_1\,h_n\,f_0(r_0)}
	\sum_{i:Z_i=1}
	K\!\Big(\frac{R_{i0}-r_0}{h_n}\Big)\,\eta_i,
	\]
	\[
	\begin{aligned}
		\mathsf{S}_\phi(r,r_0)
		&=
		\frac{f_0(r_0)}{f^{cf}(r)}\cdot \frac{1}{n_1}\sum_{i:Z_i=1}\Psi_{1|0,i}\!\big(r-\Delta c,\,r_0-\Delta c\big)
		\;+\;
		\frac{f_{1|0}(r-\Delta c\,|\,r_0-\Delta c)}{f^{cf}(r)}\cdot \frac{1}{n_1}\sum_{i:Z_i=1}\Xi^{(0)}_i(r_0)
		\\[-0.25em]
		&\hspace{6em}
		-\;
		\frac{f_{1|0}(r-\Delta c\,|\,r_0-\Delta c)\,f_0(r_0)}{f^{cf}(r)^2}\cdot \frac{1}{n_1}\sum_{i:Z_i=1}\kappa_i(r).
	\end{aligned}
	\]
\end{lem}

\begin{lem}[Change of variables and second–order expansion]\label{lem: change of variable integration }
	Let $\{(R_i,\varepsilon_i)\}_{i=1}^n$ be i.i.d.\ with $\mathbb E[\varepsilon_i\mid R_i]=0$ and $\mathbb E[|\varepsilon_i|^{2+\delta}]<\infty$ for some $\delta>0$. Let $w:\mathbb R\to\mathbb R$ be twice continuously differentiable with bounded $w''$. 
	Define
	\[
	T_n \;:=\; \int_{c^{cf}}^\infty \frac{1}{n h_n}\sum_{i=1}^n K\!\left(\frac{R_i - r}{h_n}\right)\varepsilon_i\, w(r)\,dr,
	\]
	where $h_n\to 0$ as $n\to\infty$. 
	Then
	\[
	T_n = \frac{1}{n}\sum_{i=1}^n \varepsilon_i w(R_i) F_K\left(\frac{R_i-c^{cf}}{h_n}\right)
	+ O_p(h_n^2),
	\]
	where $F_K(t)=\int_{-\infty}^t K(u) du$. In particular, if $n^{\alpha}\,h_n^2 \to 0$, then
	\[
	{n}^{\alpha}\,T_n \;\Rightarrow\; \mathcal N\!\left(0,\; \mathrm{Var}\left(w(R_i)F_K\left(\frac{R_i-c^{cf}}{h_n}\right)\varepsilon_i\right)n^{\alpha-1/2}\right).
	\]
\end{lem}

\begin{proof}
	We first interchange the integral and summation to write
	\[
	T_n = \frac{1}{n}\sum_{i=1}^n \varepsilon_i \int_{r\ge c^{cf}} \frac{1}{h_n} K\!\left(\frac{R_i - r}{h_n}\right) w(r)\,dr.
	\]
	Make the change of variable $u = (R_i - r)/h_n$, so $r = R_i - h_n u$ and $dr = -h_n du$. 
	Then
	\[
	\int_{r\ge c^{cf}} \frac{1}{h_n} K\!\left(\frac{R_i - r}{h_n}\right) w(r)\,dr 
	= \int_{u\ge (R_{i}-c^{cf})/h_n} K(u)\, w(R_i - h_n u)\,du.
	\]
	Hence
	\[
	T_n = \frac{1}{n}\sum_{i=1}^n \varepsilon_i \int_{u\ge (R_{i}-c^{cf})/h_n} K(u)\, w(R_i - h_n u)\,du.
	\]
	
	Next, expand $w(R_i - h_n u)$ in a second–order Taylor expansion around $R_i$:
	\[
	w(R_i - h_n u) 
	= w(R_i) - h_n u\,w'(R_i) + \tfrac{1}{2} h_n^2 u^2\, w''(\xi_{i,u}),
	\]
	for some $\xi_{i,u}$ between $R_i$ and $R_i - h_n u$. 
	
	Now, we can write 
	\[
	\begin{split}
		T_n &= \frac{1}{n}\sum_{i=1}^n \varepsilon_i \int_{u\ge (R_{i}-c^{cf})/h_n} K(u)\, [ w(R_i) - h_n u\,w'(R_i) + \tfrac{1}{2} h_n^2 u^2\, w''(\xi_{i,u})]\,du\\
		&=\frac{1}{n}\sum_{i=1}^n \varepsilon_i w(R_i) F_K\left(\frac{R_i-c^{cf}}{h_n}\right) \\
		&-  \frac{1}{n}\sum_{i=1}^n \varepsilon_i \int_{u\ge (R_{i}-c^{cf})/h_n} K(u)\,   h_n u\,w'(R_i) \,du +O_p(h_n^2)
	\end{split}
	\]
	Note that for $R_i\notin [c^{cf}-h_n,c^{cf}+h_n]$, we know that $ \int_{u\ge (R_{i}-c^{cf})/h_n} K(u) u du=0$, so we can have 
	\[
	 \frac{1}{n}\sum_{i=1}^n \varepsilon_i \int_{u\ge (R_{i}-c^{cf})/h_n} K(u)\,   h_n u\,w'(R_i) \,du =\frac{h_n}{n}\sum_{i=1}^n \varepsilon_i F_K\left(\frac{R_i-c^{cf}}{h_n}\right)\mathbbm{1}(R_{i}\in  [c^{cf}-h_n,c^{cf}+h_n]).
	\]
	Let $n'=\sum_{i}\mathbbm{1}(R_{i}\in  [c^{cf}-h_n,c^{cf}+h_n]$. Since $nh_n\rightarrow \infty$, $n'\rightarrow \infty$ with probability approaching 1, then we have 
	\[
	\frac{1}{n'}\sum_{i=1}^n \varepsilon_i F_K\left(\frac{R_i-c^{cf}}{h_n}\right)\mathbbm{1}(R_{i}\in  [c^{cf}-h_n,c^{cf}+h_n])= O_p(1),
	\]
	and $n'/n\rightarrow_p f^{cf}(c^{cf})h_n$, so we can show that $ \frac{1}{n}\sum_{i=1}^n \varepsilon_i \int_{u\ge (R_{i}-c^{cf})/h_n} K(u)\,   h_n u\,w'(R_i) \,du=O_p(h_n^2).$

	Hence
		\[
	T_n = \frac{1}{n}\sum_{i=1}^n \varepsilon_i w(R_i) F_K\left(\frac{R_i-c^{cf}}{h_n}\right)
	+ O_p(h_n^2).
	\]
	
	For the limiting distribution, note that $\mathbb E[\varepsilon_i w(R_i)] = 0$ by mean independence. 
	Because $\mathbb E[\varepsilon_i^2 w(R_i)^2]<\infty$, the Lindeberg–Feller CLT gives
	\[
	\sqrt{n}\,\frac{1}{n}\sum_{i=1}^n \varepsilon_i w(R_i) F_K\left(\frac{R_i-c^{cf}}{h_n}\right)
	\;\Rightarrow\; \mathcal N(0,V),
	\qquad 
	V=\mathrm{Var}\left(w(R_i)\varepsilon_i F_K\left(\frac{R_i-c^{cf}}{h_n}\right)\right).
	\]
	The second–order term is $O_p(h_n^2)$ and the remainder $o_p(h_n^2)$, both $o_p(n^{-1/2})$ if $n^{\alpha}\,h_n^2\to0$. 
	Therefore, by Slutsky’s theorem,
	\[
	n^{\alpha}\,T_n \;\Rightarrow\; \mathcal N(0,V*n^{\alpha-1/2}),
	\]
	as claimed.
\end{proof}

\begin{lem}\label{lemma: linearization for the second term in T}
	For $h_n=b_n=n^{-1/5-\delta}$ and $\delta>0$ is a small number, then we have 
	{\footnotesize
		\[
		\begin{split}
			&\quad \sqrt{nb_n} \int_{r_0}  \hat\theta(r_0)\,\hat\phi(r,r_0)-\theta(r_0)\,\phi(r,r_0) dr_0\\
			&= \frac{1}{\alpha\sqrt{nb_n}}\frac{\int_{r_0} f_{1|0}(r-\Delta c|r_0-\Delta c)f_0(r_0)dr_0}{f^{cf}(r)^2} \sum_{i:Z_i=1}  \frac{\Big[K_{b_n}(R_{i1}-(r-\Delta c))-E\left[K_{b_n}(R_{i1}-(r-\Delta c))|R_{i0}-\Delta c\right]\Big]\,f_0(R_{i0}+\Delta c )}{f_{0|1}(R_{i0})}  \\
			&+ \frac{1}{\alpha\sqrt{nb_n}}\sum_{Z_i=1} \theta(R_{i0}+\Delta c)\frac{f_0(R_{i0}+\Delta c)}{f^{cf}(r)f_{0|1}(R_{i0})}  \left[K\left(\frac{R_{i1}-r+\Delta c}{b_n}\right)-E\left[K\left(\frac{R_{i1}-r+\Delta c}{b_n}\right)\mid R_{i0}\right]\right] \\
			&+O_p(n^{-2.5\delta}+n^{2/5-1/2-\delta/2} +   n^{-0.3+\delta}\log n  +n^{-\delta/2})
		\end{split}
		\]}
\end{lem}

\begin{proof}
	First, we have $ \sqrt{nb_n} =n^{2/5-\delta/2}$. Then, by Lemma \ref{lem:theta-phi-product}, it is easy to see that 
	\[\begin{split}
		&\quad	\sqrt{nb_n}\int_{r_0}\left| \hat\theta(r_0)\,\hat\phi(r,r_0)-\theta(r_0)\,\phi(r,r_0) - ( \phi(r,r_0)\,\mathsf{S}_\theta(r_0)
		\;+\; \theta(r_0)\,\mathsf{S}_\phi(r,r_0))dr_0 \right|\\
		&= O_p\left(n^{2/5-\delta/2}(h_n)^2+ n^{2/5-\delta/2}(b_n)^2 +n^{2/5-\delta/2}\frac{\log n}{ n h_n \sqrt{b_n } }\right)\\
		&= O_p(n^{-2.5\delta} +  n^{-0.3+\delta} \log n ).
	\end{split}
	\]
	
	Then using Lemma \ref{lem: change of variable integration }, we can derive 
	\[
	\sqrt{nb_n}\int_{r_0} \phi(r,r_0) S_\theta(r_0) dr_0=O_p(n^{2/5-1/2-\delta/2}).
	\]
	For the $\int_{r_0} S_{\phi}(r,r_0)dr_0$, we use the change of variable $(R_{i0}-r_0)/h_n= u$
	\[
	\begin{split}
		&\quad 	\int_{r_0}
		\frac{f_{1|0}(r-\Delta c\,|\,r_0-\Delta c)}{f^{cf}(r)}\cdot \frac{1}{n_1}\sum_{i:Z_i=1}\Xi^{(0)}_i(r_0)\\
		&=_{(1)}	\int_{u}
		\frac{1}{n_1}\sum_{i:Z_i=1} \frac{f_{1|0}(r-\Delta c\,|\,R_{i0}-uh_n-\Delta c)}{f^{cf}(r)}\cdot K(u) du - \int_{r_0} 	\frac{f_{1|0}(r-\Delta c\,|\,r_0-\Delta c)}{f^{cf}(r)}\cdot f_0(r_0) + O(h_n^2) dr_0\\
		&=_{(2)}\frac{1}{n_1}\sum_{i:Z_i=1}	\frac{f_{1|0}(r-\Delta c\,|\,R_{i0}-\Delta c)}{f^{cf}(r)}-\int_{r_0} 	\frac{f_{1|0}(r-\Delta c\,|\,r_0-\Delta c)}{f^{cf}(r)}\cdot  f_0(r_0) dr_0 +O(h_n^2)\\
		&= O_p(1/\sqrt{n} ) +O(h_n^2),
	\end{split}
	\]
	where in the $(1)$ we use the change of variable $(R_{i0}-r_0)/h_n= u$, and the fact that $E\left[K\left(\frac{R_{i0}-r_0}{h_n}\right)\right]-f_0(r_0) = O(h_n^2)$; in the $(2)$ we use the linear expansion of $uh_n$ to get another bias term of order $O(h_n^2)$, which is absorbed in the remainder term, and in the last step we use  the WLLN to get the $O_p(1/\sqrt{n})$ for the leading term. So for this term we have $\sqrt{nb_n}\int_{r_0} S_{\phi}(r,r_0)dr_0 = O_p(\sqrt{nb_n}h_n^2 + \sqrt{b_n}) =O_p(n^{-\delta/2})$.

	For the $\mu_i(r)$ term in $\int_{r_0} S_{\phi}(r,r_0)dr_0$, we notice that 
	\[
	\int_{r_0} 	\frac{f_{1|0}(r-\Delta c\,|\,r_0-\Delta c)\,f_0(r_0)}{f^{cf}(r)^2} \mu_i(r) dr_0 =\frac{\mu_i(r)}{f^{cf}(r)^2}\int_{r_0} f_{1|0}(r-\Delta c|r_0-\Delta c)f_0(r_0)dr_0 
	\]
	where we recall the definition of $\mu_i(r)$ here:
	\[
	\begin{split}
		\mu_i(r)
		&=  \frac{\Big[K_{b_n}(R_{i1}-(r-\Delta c))-E\left[K_{b_n}(R_{i1}-(r-\Delta c))|R_{i0}\right]\Big]\,f_0(R_{i0}+\Delta c )}{b_nf_{0|1}(R_{i0})} \\
		&+\left[f_{1|0}(r-\Delta c\mid R_{i0}-\Delta c)- \int_{r_0} f_0(r_0)f_{1|0}(r-\Delta c|r_0-\Delta c) dr_0\right].
	\end{split}
	\]
	We note that for the second summand of $\mu_i(r)$, we have 
	\[
	\frac{1}{n_1} \sum_{Z_i=1} \left[f_{1|0}(r_1-\Delta c\mid R_{i0}-\Delta c)- \int_{r_0} f_0(r_0)f_{1|0}(r_1-\Delta c|r_0-\Delta c) dr_0\right] =O_p(1/\sqrt{n_1}),
	\]
	so 
	{\footnotesize
		\[
		\begin{split}
			&\quad \frac{\sqrt{nb_n}}{n_1}\sum_{i:Z_i=1}\frac{\mu_i(r)}{f^{cf}(r)^2}\int_{r_0} f_{1|0}(r-\Delta c|r_0-\Delta c)f_0(r_0)dr_0\\& = \frac{1}{\alpha\sqrt{nb_n}}\frac{\int_{r_0} f_{1|0}(r-\Delta c|r_0-\Delta c)f_0(r_0)dr_0}{f^{cf}(r)^2} \sum_{i:Z_i=1}  \frac{\Big[K_{b_n}(R_{i1}-(r-\Delta c))-E\left[K_{b_n}(R_{i1}-(r-\Delta c))|R_{i0}\right]\Big]\,f_0(R_{i0}+\Delta c )}{f_{0|1}(R_{i0})} 
		\end{split}
		\]
	}

	So it remains to derive the linear expansion of 
	\[
	I_n(r)\equiv \int_{r_0}\left[\frac{\theta(r_0)f_0(r_0)}{f^{cf}(r)}\cdot \frac{1}{n_1}\sum_{i:Z_i=1}\Psi_{1|0,i}\!\big(r-\Delta c,\,r_0-\Delta c\big)\right] dr_0.
	\]
	First, recall that $x=r_0-\Delta c$ and $y=r-\Delta c$, we can write
	{\footnotesize
		\[
		\Psi_{1|0,i}(y,x)
		:= 
		\frac{K\!\left(\dfrac{R_{i0}-x}{h_n}\right)}{h_n\,f_{0}(x)}
		\left[
		\frac{K\!\left(\dfrac{R_{i1}-y}{b_n}\right)}{b_n}
		- \frac{E\left[{K\!\left(\dfrac{R_{i1}-y}{b_n}\right)}\mid R_{i0}=x\right]}{b_n}+\underbrace{\frac{E\left[{K\!\left(\dfrac{R_{i1}-y}{b_n}\right)}\mid R_{i0}=x\right]}{b_n}-  f_{1|0}(y\mid x)}_{bias}
		\right], 
		\]}
	where the bias term is of order $O(b_n^2)$, uniformly over $(x,y)$ values, by standard linearization argument and the symmetry of kernel function. Therefore, we can write 
	
		{\footnotesize
		\[
		\begin{split}
			n_1 I_n(r)&= \int_{r_0}\frac{\theta(r_0)f_0(r_0)}{n_1f^{cf}(r)}\sum_{i:Z_i=1}\frac{K\!\left(\dfrac{R_{i0}-r_0+\Delta c}{h_n}\right)}{h_n\,f_{0|1}(r_0-\Delta c)}
			\left[
			\frac{K\!\left(\dfrac{R_{i1}-r+\Delta c}{b_n}\right)}{b_n}
			-\underbrace{\frac{E\left[{K\!\left(\dfrac{R_{i1}-r+\Delta c}{b_n}\right)}\mid R_{i0}=r_0-\Delta c\right]}{b_n}}_{:=\Gamma(r_0-\Delta c)}
			\right] dr_0 +O(b_n^2)\\
			&=_{(i)}\frac{1}{f^{cf}(r)}
			\frac{K\!\left(\dfrac{R_{i1}-r+\Delta c}{b_n}\right)}{b_n} \times \Bigg\{ \sum_{i:Z_i=1} \int_{u}\frac{\theta(R_{i0}+\Delta c-uh_n)f_0(R_{i0}+\Delta c-uh_n)}{n_1}\frac{K(u)}{f_{0|1}(R_{i0}-uh_n)} du \Bigg\}+O(b_n^2)\\
			&- \sum_{i:Z_i=1} \int_{u}\frac{\theta(R_{i0}+\Delta c-uh_n)f_0(R_{i0}+\Delta c-uh_n)}{n_1}\frac{K(u)}{f_{0|1}(R_{i0}-uh_n)} \Gamma(R_{i0}-uh_n)du  \\
			&=_{(ii)}\frac{1}{f^{cf}(r)}\times\sum_{i:Z_i=1} \int_{u}\frac{\theta(R_{i0}+\Delta c)f_0(R_{i0}+\Delta c)}{n_1}\frac{K(u)}{f_{0|1}(R_{i0})} du \times \left\{ \frac{K\!\left(\dfrac{R_{i1}-r+\Delta c}{b_n}\right)}{b_n} - \Gamma(R_{i0})  \right\} +O(h_n^2).
		\end{split}
		\]
	}
	Rewrite the equation we can get 
	
	{\footnotesize
	\[
	\begin{split}
		I_n(r)= \frac{1}{n_1b_n}\sum_{Z_i=1} \theta(R_{i0}+\Delta c)\frac{f_0(R_{i0}+\Delta c)}{f^{cf}(r)f_{0|1}(R_{i0})}  \left[K\left(\frac{R_{i1}-r+\Delta c}{b_n}\right)-E\left[K\left(\frac{R_{i1}-r+\Delta c}{b_n}\right)\mid R_{i0}\right]\right] +O(h_n^2+b_n^2),
	\end{split}
	\]
}
In  step $(i)$ we use the change of variable $u=\dfrac{R_{i0}-r_0+\Delta c}{h_n}$ and Fubini's theorem, in step $(ii)$ we use the local expansion of the $h_n$. Then multiply $\sqrt{nb_n}$ on both sides of $I_n(r)$, and use $\sqrt{nb_n}h_n^2\asymp \sqrt{nb_n}b_n^2\asymp n^{-2.5\delta}$ to get the result in the lemma. 
\end{proof}

\begin{lem}[Uniform linearization of the boundary-kernel estimator $\widehat m_1(r)$]\label{lem:m1-bdry-unif}
	Let $\widehat m_1(r)$ be the boundary-aware estimator in (3.5) with bandwidth $b_n$,
	and define $m_1(r):=\mathbb E[Y_{i1}\mid R_{i1}=r,Z_i=1]$, $f_1(r):=f_{R_{i1}\mid Z=1}(r)$, and
	$\eta_{i1}:=Y_{i1}-m_1(R_{i1})$ on $\{Z=1\}$. For $r\in\mathcal I\subseteq[c,\bar r]$ compact, write
	\[
	K_r^*(u)=
	\begin{cases}
		K(u), & r-c\ge \tau b_n,\\
		K^{\mathrm{bd}}_{r,b_n}(u), & c\le r-c< \tau b_n,
	\end{cases}
	\quad
	\kappa_2(r):=\int u^2 K_r^*(u)\,du,\qquad R(K_r^*):=\int (K_r^*(u))^2\,du.
	\]
	Then,
	\[
	\sup_{r\in\mathcal I}\Big|
	\widehat m_1(r)-m_1(r)
	-\frac{1}{n_1 b_n f_1(r)}
	\sum_{i:Z_i=1} K_r^*\!\Big(\frac{R_{i1}-r}{b_n}\Big)\,\eta_{i1}
	- B_1^{\mathrm{bd}}(r)
	\Big|
	=O_p\!\Big(\sqrt{\tfrac{\log n_1}{n_1 b_n}}\Big)
	+o\!\big((b_n)^2\big),
	\]
	with local-constant bias
	\[
	B_1^{\mathrm{bd}}(r)=O\left((b_n)^2\right).
	\]
\end{lem}

\begin{proof}
	The linearization of conditional mean under the boundary kernel estimation is standard in the literature and can be found in the textbook, see Section 3.2-3.4 of \cite{FanGijbels1996}. The stochastic order term has a slightly different order $\tfrac{\log n_1}{n_1 b_n}$ with an additional $\log n_1$ term because of the uniform linearization.
\end{proof}

\begin{lem}[Uniform linearization of $\,\hat f_{1|0}(r-\Delta c \mid r_0-\Delta c)\,\hat f_0(r_0)$]
	Fix compact sets $I,I_0\subset\mathbb R$, let $r\in I$, $r_0\in I_0$, and set $y:=r-\Delta c$, $x:=r_0-\Delta c$.
	Define
	\[
	\widehat{\psi}(r,r_0):=\hat f_{1|0}(y\mid x)\,\hat f_0(r_0),\qquad
	\psi(r,r_0):= f_{1|0}(y\mid x)\,f_0(r_0).
	\]
	Then, uniformly over $(r,r_0)\in I\times I_0$,
	\begin{align*}
		\sup_{r,r_0}\ \Bigg|
		&\widehat{\psi}(r,r_0)-\psi(r,r_0)\\
		&	-\Bigg[
		f_0(r_0)\cdot\frac{1}{n_1}\sum_{i:Z_i=1}\Psi_{1|0,i}(y,x)
		+ f_{1|0}(y\mid x)\cdot\frac{1}{n_1}\sum_{i:Z_i=1}\tilde \Xi_i^{(0)}(r_0) \Bigg]\Bigg|
		\\
		&=O_p\!\left(h_n^{2}+b_n^{2}+\frac{\log n}{nh_n\sqrt{b_n}}
		\right),
	\end{align*}
	where 
		\[
	\tilde\Xi^{(0)}_{i}(r_0)
	:=
	\frac{ K\!\big(\frac{R_{i0}-r_0}{h_n}\big) }{h_n}
	-f_0(r_0),
	\]
\end{lem}

\begin{lem}[Uniform linearization of $\ \hat\theta(r_0)\,\hat\psi(r,r_0)$]
	Fix compact sets $I,I_0\subset\mathbb R$, let $r\in I$, $r_0\in I_0$, and set $y:=r-\Delta c$, $x:=r_0-\Delta c$.
	Then, uniformly over $(r,r_0)\in I\times I_0$,
	{
		\begin{align*}
				&\sup_{r,r_0}\Bigg|
				\hat\theta(r_0)\hat\psi(r,r_0)-\theta(r_0)\psi(r,r_0)-\Bigg[
				\psi(r,r_0)\mathsf{S}_{\theta}(r_0)+ \theta(r_0)\!\left\{\frac{f_0(r_0)}{n_1}\sum_{i:Z_i=1}	\Psi_{1|0,i}(y,x) +\frac{f_{1|0}(y,x)}{n_1} \sum_{i:Z_i=1}\tilde \Xi_i^{(0)}(r_0)
				 \right\}
				\Bigg]\Bigg|\\
				&\qquad=O_p\!\left(
				h_n^{2}+b_n^{2}+\frac{\log n}{nh_n\sqrt{b_n}}
				\right).
		\end{align*}}
\end{lem}

\begin{lem}\label{lem: linearization for the ATT integrant}
	Recall that $\tilde{f}_0(r_0)$ is the density of $R_{i0}$ conditional on the $Z_i=0$ group. Let $F_{K,i}$ denote the random variable 
	$$F_{K,i}:=F_K\!\left( \frac{R_{i1} - (c^{cf} - \Delta c)}{b_n} \right)=\int^{u\ge\frac{R_{i1} - (c^{cf} - \Delta c)}{b_n} }_{-\infty} K(u) du.$$

	 For $h_n=b_n=n^{-1/4-\delta}$ and $\delta>0$ is a small number, then we have 
	\[
	\begin{split}
		&\quad \sqrt{n} \int_{r=c^{cf}}^{\infty}  \int_{r_0} \hat\theta(r_0)\,\hat\psi(r,r_0)-\theta(r_0)\,\psi(r,r_0) dr_0 dr\\
		&= \frac{1}{(1-\alpha)\sqrt{n}}\sum_{i:Z_i=0} \frac{f_0(R_{i0})}{\tilde{f}_0(R_{i0})} \varepsilon_i \kappa(R_{i0})+  \frac{1}{\alpha \sqrt{n}}\sum_{i:Z_i=1}\eta_{i0} \kappa(R_{i0})\\
		&+\frac{1}{\alpha\sqrt{n}}\sum_{i: Z_i=1}\theta(R_{i0})(1-F_{1|0}(c^{cf}-\Delta c|R_{i0}-\Delta c)-E[\theta(R_{i0})(1-F_{1|0}(c^{cf}-\Delta c|R_{i0}-\Delta c)|Z_i=1]\\
		&+ \frac{1}{\alpha\sqrt{n}} 
		\sum_{i: Z_i = 1} 
		\pi(R_{i0})
		\left[
		F_{K,i}
		- 
		\mathbb{E}\!\left\{
		F_{K,i}
		\Bigm| R_{i0}
		\right\}
		\right]+o_p(1),
	\end{split}
	\]
	where $\kappa(r_0)=(1-F_{1|0}(c^{cf}-\Delta c|R_{i0}=r_0-\Delta c))$, $F_{1|0}(a,b)=Pr(R_{i1}\le a| R_{i0}=b,Z_i=1)$, and $\pi(x) = 
	\frac{
		\theta(x + \Delta c)\, f_0(x + \Delta c)
	}{
		f_{0}(x)
	}
	$.
\end{lem}

\begin{proof}
	First, by choosing  $h_n=b_n=n^{-1/4-\delta}$, we have 
	\[
	\sqrt{n}O_p\!\left(h_n^{2}+b_n^{2}+\frac{\log n}{nh_n\sqrt{b_n}}\right)= O_p(n^{-2\delta}+\frac{\log n}{\sqrt{n^{1/4-3\delta}}})=o_p(1).
	\]
	So it suffice to consider the leading terms in Lemma \ref{lem: linearization for the ATT integrant}.
	
	\noindent\textbf{Step 1: For the influence term with $\sum_{i:{Z_i=1}} \Psi_{1|0,i}(y,x)$}.
	
	For the linear influence term $\sum_{i:{Z_i=1}} \Psi_{1|0,i}(y,x)$, we cannot directly apply Lemma \ref{lem: change of variable integration }, therefore, we need to derive the integration expression. Let 
	\[
	\tilde{I}_n\equiv \int_{r=c^{cf}}^\infty \int_{r_0}\left[\theta(r_0)f_0(r_0)\cdot \frac{1}{n_1}\sum_{i:Z_i=1}\Psi_{1|0,i}\!\big(r-\Delta c,\,r_0-\Delta c\big)\right] dr_0 dr.
	\]
	We first recall that for $y=r-\Delta c$ and $x=r_0-\Delta c$, 
	\[
	\begin{split}
			\Psi_{1|0,i}(y,x)
		&:= 
		\frac{K\!\left(\dfrac{R_{i0}-x}{h_n}\right)}{h_n\,f_{0}(x)}
		\left[
		\frac{K\!\left(\dfrac{R_{i1}-y}{b_n}\right)}{b_n}-  f_{1|0}(y\mid x)
		\right]\\
		&=\frac{K\!\left(\dfrac{R_{i0}-x}{h_n}\right)}{h_n\,f_{0}(x)}
		\left[
		\frac{K\!\left(\dfrac{R_{i1}-y}{b_n}\right)}{b_n}
		- \frac{E\left[{K\!\left(\dfrac{R_{i1}-y}{b_n}\right)}\mid R_{i0}=x\right]}{b_n}+O(b_n^2),
		\right]
	\end{split}
	\]
	where the remainder ther $O(b_n^2)$ is uniform over all $x$ value. Then we evaluate the integration over $r$:
	 {\footnotesize
\[
\int_{r=c^{cf}}^{\infty} 
\Psi_{1|0,i}(r-\Delta c, x)\,dr
=
\frac{K\!\left(\dfrac{R_{i0}-x}{h_n}\right)}{h_n f_{0}(x)}
\left[
F_K\!\left(\dfrac{R_{i1}-(c^{cf}-\Delta c)}{b_n}\right)
-
\mathbb{E}\!\left\{
F_K\!\left(\dfrac{R_{i1}-(c^{cf}-\Delta c)}{b_n}\right)
\Bigm| R_{i0}=x
\right\}
\right]
+ O(b_n^2).
\]
}
Then, define $\tau(x)=\mathbb{E}\!\left\{
F_K\!\left(\dfrac{R_{i1}-(c^{cf}-\Delta c)}{b_n}\right)
\Bigm| R_{i0}=x
\right\}$, we than take the integration over $r_0$, so that 
{\footnotesize
\[
\int \theta(r_0) f_0(r_0) 
\frac{K\!\left(\dfrac{R_{i0}-x}{h_n}\right)}{h_n f_{0}(x)}
\left[ F_K(\dfrac{R_{i1}-(c^{cf}-\Delta c)}{b_n}) - \tau(x) \right] dr_0
=
\int \pi(x)\, K(\frac{R_{i0}-x}{h_n})\, \left[ F_K(\dfrac{R_{i1}-(c^{cf}-\Delta c)}{b_n}) - \tau(x) \right] dx,
\]}
where $\pi(x) = 
\frac{
	\theta(x + \Delta c)\, f_0(x + \Delta c)
}{
	f_{0}(x)
}.
$
Next, use the change of variable, that $u=(R_{i0}-x)/h_n$, and use the second order linear expansion for $uh_n$, we can have 
\[
\begin{split}
&\quad \int \pi(x)\, K(\frac{R_{i0}-x}{h_n})\, \left[ F_K(\dfrac{R_{i1}-(c^{cf}-\Delta c)}{b_n}) - \tau(x) \right] dx\\
&= \frac{1}{n_1} 
\sum_{i: Z_i = 1} 
\pi(R_{i0})
\left[
F_K\!\left( \frac{R_{i1} - (c^{cf} - \Delta c)}{b_n} \right)
- 
\mathbb{E}\!\left\{
F_K\!\left( \frac{R_{i1} - (c^{cf} - \Delta c)}{b_n} \right)
\Bigm| R_{i0}
\right\}
\right]
+ O(h_n^2).
\end{split}
\]
So we conclude that
	\[
\tilde{I}_n=\frac{1}{n_1} 
\sum_{i: Z_i = 1} 
\pi(R_{i0})
\left[
F_{K,i}
- 
\mathbb{E}\!\left\{
F_{K,i}
\Bigm| R_{i0}
\right\}
\right]
+ O(h_n^2+b_n^2).
\]

\noindent \textbf{Step 2: For the influence term with $\sum_{i:Z_i=1}\tilde \Xi^{(0)}_i (r_0)$. }

Let
{\footnotesize
	\[
	\begin{split}
			 \breve{I}_n&=\int_{r=c^{cf}}^\infty \int_{r_0} \frac{f_{1|0}(y|x)\theta(r_0)}{n_1}\sum_{i:Z_i=1}\tilde \Xi^{(0)}_i (r_0) dr_0 dr\\
		&= \frac{1}{n_1}\sum_{i: Z_i=1}\int_{r_0}\theta(r_0) \left[\frac{ K\!\big(\frac{R_{i0}-r_0}{h_n}\big) }{h_n}-f_0(r_0)\right] (1-F_{1|0}(c^{cf}-\Delta c|R_{i0}=r_0-\Delta c))dr_0\\
		&=\frac{1}{n_1}\sum_{i: Z_i=1}\left\{\int_{r_0}\frac{\theta(r_0)}{h_n}K\!\big(\frac{R_{i0}-r_0}{h_n}\big)(1-F_{1|0}(c^{cf}-\Delta c|R_{i0}=r_0-\Delta c)dr_0 -E[\theta(R_{i0})(1-F_{1|0}(c^{cf}-\Delta c|R_{i0}-\Delta c))|Z_i=1]\right\}\\
		&= \frac{1}{n_1}\sum_{i: Z_i=1} \theta(R_{i0})(1-F_{1|0}(c^{cf}-\Delta c|R_{i0}-\Delta c)-E[\theta(R_{i0})(1-F_{1|0}(c^{cf}-\Delta c|R_{i0}-\Delta c)|Z_i=1]+O(h_n^2),
	\end{split}
	\]}
	where in the last step of derivation we use the change of variable $u=(R_{i0}-r_0)/h_n$ and use a second order linearization for the $uh_n$ term. 
	
	\noindent\textbf{Step 3: For theinfluence term with $\mathsf{S}_\theta(r_0)$}.
	
	First define $\kappa(r_0)=(1-F_{1|0}(c^{cf}-\Delta c|R_{i0}=r_0-\Delta c))$, the remaining leading terms in integral expansion are
	\[
	\begin{split}
		&\quad \int_{r\ge c^{cf}}\int_{r_0} \psi(r,r_0)\mathsf{S}_{\theta}(r_0) dr_0 dr\\
		&=\int_{r_0}\int_{r=c^{cf}}^\infty f_{1|0}(r-\Delta c| r_0- \Delta c) f_0(r_0)\mathsf{S}_{\theta}(r_0) dr dr_0\\
		&=\int_{r_0}(1-F_{1|0}(c^{cf}-\Delta c|R_{i0}=r_0-\Delta c)) f_0(r_0)\mathsf{S}_{\theta}(r_0) dr_0 \\
		&=\int_{r_0} f_0(r_0) \kappa(r_0)	\left[\frac{1}{n_0\,h_n\,\tilde{f}_0(r_0)}
		\sum_{i:Z_i=0}
		K\!\Big(\frac{R_{i0}-r_0}{h_n}\Big)\,\varepsilon_i
		\;+\;
		\frac{1}{n_1\,h_n\,f_0(r_0)}
		\sum_{i:Z_i=1}
		K\!\Big(\frac{R_{i0}-r_0}{h_n}\Big)\,\eta_i\right] dr_0\\
		&=\frac{1}{n_0}\sum_{i:Z_i=0} \varepsilon_i\kappa(R_{i0})\frac{f_{0}(R_{i0})}{\tilde{f}_0(R_{i0})} +  \sum_{i:Z_i=0}\frac{1}{n_1} \eta_i\kappa(R_{i0}) +O(h_n^2),
	\end{split}
	\]
	where the last equality holds by Lemma \ref{lem: change of variable integration } with $c^{cf}=-\infty$, and $\tilde{f}_0=f(R_{i0}|Z_i=0)$.

\end{proof}

\begin{lem}\label{lem: linearization for m1 integration for ATT}
	Using the same notation of $F_{K,i}$, we have the following linearization 
	\[
	\begin{split}
		&\quad \sqrt{n}\int_{r=c^{cf}}^{\infty} \left[ \widehat{m}_1(r) \widehat{f}^{cf}(r) dr - \int_{r=c^{cf}}^\infty {m}_1(r) {f}^{cf}(r) \right]dr\\
		&=\frac{1}{\alpha \sqrt{n}}\sum_{i:Z_i=1} \frac{\eta_{i1}f^{cf}(R_{i1})}{f_1(R_{i1})}F_{K^*}\left(\frac{R_{i1}-c^{cf}}{b_n}\right) \\
		&+\frac{1}{\alpha \sqrt{n}} \sum_{i:Z_i=1} m_1(R_{i1}+\Delta c) \left\{F_{K,i}-E\left[F_{K,i}|R_{i0}\right] \right\} \frac{f_0(R_{i0}+\Delta c)}{f_0 (R_{i0})}\\
		&+\frac{1}{\alpha \sqrt{n}} \sum_{i:Z_i=1}  \left\{\int_{r\ge c^{cf}} m_1(r)f_{1|0}(r_1-\Delta c|R_{i0}-\Delta c) dr - E\left[\int_{r\ge c^{cf}} m_1(r)f_{1|0}(r_1-\Delta c|R_{i0}-\Delta c) \right] \right\}+O_p(b_n^2+h_n^2)
	\end{split}
	\]
\end{lem}

\begin{proof}
	Following the same argument, we have 
	\[
	\begin{split}
		&\quad 	\int_{r=c^{cf}}^{\infty} \left[ \widehat{m}_1(r) \widehat{f}^{cf}(r) dr - \int_{r=c^{cf}}^\infty {m}_1(r) {f}^{cf}(r) \right]dr\\
		&= \int_{r=c^{cf}}^{\infty} \frac{1}{n_1 b_n f_1(r)}
		\sum_{i:Z_i=1} K_r^*\!\Big(\frac{R_{i1}-r}{b_n}\Big)\,\eta_{i1} f^{cf}(r) dr \\
		&+ \int_{r=c^{cf}}^{\infty} m_1(r)\frac{1}{n_1} \sum_{i:Z_i=1} \mu_i(r)  dr +O_p(h_n^2+b_n^2)
	\end{split}
	\]
	
	Using Lemma \ref{lem: change of variable integration }, we have 
	\[
	\int_{r=c^{cf}}^{\infty} \frac{1}{n_1 b_n f_1(r)}
	\sum_{i:Z_i=1} K_r^*\!\Big(\frac{R_{i1}-r}{b_n}\Big)\,\eta_{i1} f^{cf}(r) dr = \frac{1}{n_1}\sum_{i:Z_i=1} \frac{\eta_{i1}f^{cf}(R_{i1})}{f_1(R_{i1})}F_{K^*}\left(\frac{R_{i1}-c^{cf}}{b_n}\right) +O(b_n^2)
	\]
	For the second part, we use the expression of $\mu_i(r)$ to write 
	\[
	\begin{split}
		&\quad \int_{r=c^{cf}}^{\infty} m_1(r)\frac{1}{n_1} \sum_{i:Z_i=1} \mu_i(r)  dr\\
		&=\int_{r=c^{cf}}^{\infty} m_1(r)\frac{1}{n_1} \sum_{i:Z_i=1}\frac{\Big[K_{b_n}(R_{i1}-(r_1-\Delta c))
			-E\!\left[K_{b_n}(R_{i1}-(r_1-\Delta c))\mid R_{i0}\right]\Big]\,f_0(R_{i0}+\Delta c)}{b_n f_{0}(R_{i0})}dr\\
		&+\int_{r=c^{cf}}^{\infty} \frac{m_1(r)}{n_1}\sum_{i:Z_i=1}\Big[f_{1\mid 0}(r_1-\Delta c\mid R_{i0}-\Delta c)
		- \int f_0(r_0)\,f_{1\mid 0}(r_1-\Delta c\mid r_0-\Delta c)\,dr_0\Big] dr\\
		&=\frac{1}{n_1} \sum_{i:Z_i=1} m_1(R_{i1}+\Delta c) \left\{F_K\left(\frac{R_{i1}-c^{cf}+\Delta}{b_n}\right)-E\left[F_K\left(\frac{R_{i1}-c^{cf}+\Delta}{b_n}\right)|R_{i0}\right] \right\} \frac{f_0(R_{i0}+\Delta c)}{f_0 (R_{i0})}\\
		&+\frac{1}{n_1} \sum_{i:Z_i=1}  \left\{\int_{r\ge c^{cf}} m_1(r)f_{1|0}(r_1-\Delta c|R_{i0}-\Delta c) dr - E\left[\int_{r\ge c^{cf}} m_1(r)f_{1|0}(r_1-\Delta c|R_{i0}-\Delta c) \right] \right\}+O(b_n^2),
	\end{split}
	\]
	where in the last step we use change of variable $u=(R_{i1}-(r_1-\Delta c))/b_n$ to evaluate the integral over $r$.

\end{proof}

\subsubsection{Proof for Bootstrap Inference Proposition \ref{prop: bootstrap validity}}

\begin{proof}
	Let $W_{i,z}^{*,b}$ denote the resampling weights for observation $i$ in bootstrap sample $b$ for group $Z_i=z\in\{0,1\}$, such that $E[W_i^*]=1$ and $\sum_{i} W_{i,z}^{*,b}= n_z$.  Following the same logic to derive Theorem \ref{thm: root-n consistency of the ATT(r)}, we can derive the linear expansion of the bootstrap estimator
	
	{\footnotesize
		\[
		\begin{split}
			&\quad \sqrt{n}\,F^{cf}(c^{cf})\left({ATT}^{*,b}(c^{cf})-\widehat{ATT}(c^{cf})\right)\\
			&= \frac{1}{\alpha \sqrt{n}}\sum_{i:Z_i=1} 
			\frac{\hat{\eta}_{i1}\,\hat f^{cf}\!\left(R_{i1}\right)}{\hat f_1\!\left(R_{i1}\right)}
			\,F_{K^*}\!\left(\frac{R_{i1}-c^{cf}}{b_n}\right)
			\,(W_{i,1}^{*,b}-1)\\
			&\quad+\frac{1}{\alpha \sqrt{n}} \sum_{i:Z_i=1} 
			\hat m_1\!\left(R_{i1}+\Delta c\right)
			\left\{
			F_{K,i}-\mathbb{E}\!\left[F_{K,i}\mid R_{i0}\right] 
			\right\}
			\frac{\hat f_0\!\left(R_{i0}+\Delta c\right)}{\hat f_0\!\left(R_{i0}\right)}
			\,(W_{i,1}^{*,b}-1)
			\\
			&\quad+\frac{1}{\alpha \sqrt{n}}  \sum_{i:Z_i=1}  
			\Bigg\{
			\int_{r\ge c^{cf}} \hat m_1(r)\,
			\widehat{f}_{1|0}\!\left(r-\Delta c\,\bigm|\,R_{i0}-\Delta c\right) \,dr 
			- \mathbb{E}\!\left[\int_{r\ge c^{cf}} \hat m_1(r)\,
			\widehat{f}_{1|0}\!\left(r-\Delta c\,\bigm|\,R_{i0}-\Delta c\right) dr \right]
			\Bigg\}
			\,(W_{i,1}^{*,b}-1)\\
			&\quad-\frac{1}{(1-\alpha)\sqrt{n}}\sum_{i:Z_i=0} 
			\frac{\hat f_0\!\left(R_{i0}\right)}{\widehat{\tilde f}_0\!\left(R_{i0}\right)} 
			\,\hat{\varepsilon}_i\, \hat{\kappa}\!\left(R_{i0}\right)
			\,(W_{i,0}^{*,b}-1)
			\;-\;  \frac{1}{\alpha \sqrt{n}}\sum_{i:Z_i=1}\hat{\eta}_{i0} \,\hat{\kappa}\!\left(R_{i0}\right)
			\,(W_{i,1}^{*,b}-1)\\
			&\quad-\frac{1}{\alpha\sqrt{n}}\sum_{i: Z_i=1} 
			\Big[
			\hat{\theta}\!\left(R_{i0}\right)\big(1-F_{1|0}\!\left(c^{cf}-\Delta c\,\bigm|\,R_{i0}-\Delta c\right)\big)
			-\mathbb{E}\!\big[\hat{\theta}\!\left(R_{i0}\right)\big(1-F_{1|0}\!\left(c^{cf}-\Delta c\,\bigm|\,R_{i0}-\Delta c\right)\big)\,\big|\,Z_i=1\big]
			\Big]
			\,(W_{i,1}^{*,b}-1)\\
			&\quad- \frac{1}{\alpha\sqrt{n}} 
			\sum_{i: Z_i = 1} 
			\hat \pi\!\left(R_{i0}\right)
			\left[
			F_K\!\left( \frac{R_{i1} - (c^{cf} - \Delta c)}{b_n} \right)
			- 
			\mathbb{E}\!\left\{
			F_K\!\left( \frac{R_{i1} - (c^{cf} - \Delta c)}{b_n} \right)
			\Bigm|\, R_{i0}
			\right\}
			\right]
			\,(W_{i,1}^{*,b}-1)
			+o_p(1),
		\end{split}
		\]
	}
	where $\hat{\eta}_{i1}=Y_{i1}-\hat{m}_1(R_{i1})$,  $\hat{\eta}_{i1}=Y_{i0}-\hat{m}_0(R_{i1})$, $\hat{\varepsilon}_i= Y_{i1}-Y_{i0}-\widehat{E}[Y_{i1}-Y_{i0}|R_{i0},Z_i=0]$.  The convergence result then follows by applying Lemma \ref{lem: bootstrap consistency lemma }.
\end{proof}

\begin{lem}\label{lem: bootstrap consistency lemma }
	Let $\{\varepsilon_i\}_{i=1}^n$ be i.i.d.\ with $\mathbb E[\varepsilon_i]=0$ and $\mathbb E|\varepsilon_i|^{2+\delta}<\infty$ for some $\delta>0$. Assume
	\[
	\sqrt n\,(\widehat{\gamma}-\gamma)=\frac1{\sqrt n}\sum_{i=1}^n \varepsilon_i + o_p(1),
	\]
	and, conditionally on the data, $(W_1^*,\dots,W_n^*)$ are multinomial bootstrap weights with $\sum_i W_i^*=n$ and
	\[
	\sqrt n\,(\gamma^*-\widehat{\gamma})=\frac1{\sqrt n}\sum_{i=1}^n \hat\varepsilon_i\,(W_i^*-1).
	\]
	If $\max_{1\le i\le n}|\hat\varepsilon_i-\varepsilon_i|\xrightarrow{p}0$, then
	\[
	\sqrt n\,(\widehat{\gamma}-\gamma)\Rightarrow \mathcal N(0,\sigma^2)
	\quad\text{and}\quad
	\sqrt n\,(\gamma^*-\widehat{\gamma})\Rightarrow{d^*}\mathcal N(0,\sigma^2),
	\]
	with $\sigma^2=\mathbb E[\varepsilon^2]$. Hence the bootstrap and original limits coincide.
\end{lem}

\begin{proof}
	Decompose
	\[
	T_n^*:=\frac1{\sqrt n}\sum_{i=1}^n \hat\varepsilon_i\,(W_i^*-1)
	= A_n^* + R_n^*,\quad
	A_n^*:=\frac1{\sqrt n}\sum_{i=1}^n \varepsilon_i\,(W_i^*-1),\quad
	R_n^*:=\frac1{\sqrt n}\sum_{i=1}^n (\hat\varepsilon_i-\varepsilon_i)\,(W_i^*-1).
	\]
	Conditional on the data, $\mathbb E^*(W_i^*-1)=0$ and $Var^*(W_i^*-1)=1-1/n$, hence
	\[
	Var^*(R_n^*)=\frac{1}{n}\sum_{i=1}^n (\hat\varepsilon_i-\varepsilon_i)^2(1-1/n)
	\le \Big(\max_{i}|\hat\varepsilon_i-\varepsilon_i|\Big)^2 \xrightarrow{p} 0,
	\]
	so $R_n^*=o_p^*(1)$. Moreover,
	\[
	Var^*(A_n^*)=\frac{1}{n}\sum_{i=1}^n \varepsilon_i^2(1-1/n)\xrightarrow{p}\sigma^2:=\mathbb E[\varepsilon^2].
	\]
	By the conditional Lyapunov CLT for exchangeably weighted sums (using $\mathbb E|\varepsilon|^{2+\delta}<\infty$),
	$A_n^*\Rightarrow{d^*}\mathcal N(0,\sigma^2)$. Therefore
	$T_n^*=A_n^*+o_p^*(1)\Rightarrow{d^*}\mathcal N(0,\sigma^2)$.
	Finally, by the ordinary CLT, $\frac1{\sqrt n}\sum_{i=1}^n \varepsilon_i \Rightarrow \mathcal N(0,\sigma^2)$, hence
	$\sqrt n(\widehat{\gamma}-\gamma)\Rightarrow \mathcal N(0,\sigma^2)$. The two limit distributions agree.
\end{proof}

\bibliographystyle{te}      
\bibliography{RDD_Extrapolation_with_policy}

\end{document}